\documentclass[a4paper,UKenglish,cleveref,autoref,thm-restate]{lipics-v2021}
\usepackage{xparse}
\usepackage{mathtools}
\usepackage{xspace}
\usepackage{cleveref}
\usepackage{graphicx}
\usepackage{tikz}
\usetikzlibrary{arrows,arrows.meta,shapes,automata,calc,chains,matrix,trees,fit,positioning,scopes,patterns,snakes}
\usetikzlibrary{decorations.pathmorphing}
\usepackage{dcolumn}
\usepackage{stackengine}
\usepackage{bm}
\usepackage[inline]{enumitem}

\usepackage{soul}
\setul{1.5pt}{0.8pt}
\setulcolor{lipicsYellow!70!white}
\usepackage{stmaryrd}
\usepackage{colonequals}
\usepackage{subcaption}

\author{Radosław Pi{\'{o}}rkowski}{University of Warwick}{radoslaw.piorkowski@warwick.ac.uk}{}{}

\authorrunning{R. Pi{\'{o}}rkowski} 

\Copyright{Radoslaw Pi{\'{o}}rkowski} 



\category{} 

\relatedversion{} 

\EventEditors{John Q. Open and Joan R. Access}
\EventNoEds{2}
\EventLongTitle{42nd Conference on Very Important Topics (CVIT 2016)}
\EventShortTitle{CVIT 2016}
\EventAcronym{CVIT}
\EventYear{2016}
\EventDate{December 24--27, 2016}
\EventLocation{Little Whinging, United Kingdom}
\EventLogo{}
\SeriesVolume{42}
\ArticleNo{23}
\newtheorem{condition}[theorem]{Condition}
\newtheorem{fact}[theorem]{Fact}

\crefname{premise}{premise}{premises}
\Crefname{premise}{Premise}{Premises}
\crefformat{premise}{#2#1#3}
\Crefformat{premise}{#2#1#3}
\crefrangeformat{premise}{#3#1#4--#5#2#6}
\crefmultiformat{premise}{#2#1#3}{ and #2#1#3}{, #2#1#3}{, and #2#1#3}

\crefname{claim}{Claim}{Claims}
\Crefname{claim}{Claim}{Claims}

\crefformat{equation}{(#2#1#3)}
\crefrangeformat{equation}{(#3#1#4) to~(#5#2#6)}
\crefmultiformat{equation}{(#2#1#3)}{ and~(#2#1#3)}{, (#2#1#3)}{ and~(#2#1#3)}

\newcommand{\of}[1]{Proof of {#1}.}

\newcommand{\automataClass}[1]{\textrm{#1}\xspace}

\newcommand{\NFA}{\parametrisedT{NFA}}
\newcommand{\DFA}{\parametrisedT{DFA}}
\newcommand{\NRA}{\parametrisedT{NRA}^{\!\textup{\textsc{g}}}}
\newcommand{\NRAweak}{\parametrisedT{NRA}^{\!\textup{\textsc{w}}}}
\newcommand{\NRAnoguess}{\parametrisedT{NRA}}

\newcommand{\NFAs}{\automataClass{NFA}}

\newcommand{\NRAs}{\automataClass{NRA${}^{\!\textup{\textsc{g}}}$}}

\newcommand{\MSO}[1][\Sigma]{\ensuremath{\textrm{MSO}[<,#1]}\xspace}
\newcommand{\MSOnoalpha}{\ensuremath{\textrm{MSO}[<]}\xspace}

\newcommand{\SMSO}[1][\Sigma, \A]{\ensuremath{\textrm{MSO}_\textrm{S}{[<,#1]}}\xspace}
\newcommand{\SMSOnoalpha}{\ensuremath{\textrm{MSO}_\textrm{S}{[<]}}\xspace}
\newcommand{\SMSOnoalphabold}{\ensuremath{\textrm{MSO}_\textrm{S}{\bm{[<]}}}\xspace}
\newcommand{\SMSOnoguess}[1][\Sigma, \A]{\ensuremath{\textrm{MSO}'_\textrm{S}{[<,#1]}}\xspace}

\newcommand{\RE}{\parametrisedT{RE}}
\newcommand{\DRE}{\parametrisedT{DRE}}
\newcommand{\DREs}{\automataClass{DRE}}

\NewDocumentEnvironment{oldtext}{O{} +b}{%
    \if\relax\detokenize{#1}\relax
    {\color{gray}\footnotesize #2}%
    \else
            {\color{blue}#1}%
    \fi
}{}
\NewDocumentEnvironment{hide}{ +b }{}{}

\newcommand{\proofCase}[2]{\smallskip\par\noindent{\color{lipicsGray}\textsf{\textbf{\hbox{Case #1}} (\hbox{\color{black}#2})}\textbf{.}}\textsf{\textbf{\hspace{0.5em}}}\ignorespaces}

\newcommand{\proofDir}[1]{\par\noindent{\color{lipicsGray}\textsf{\textbf{#1}}\textsf{\textbf{\hspace{0.5em}}}}\ignorespaces}
\newcommand{\namedparagraph}[1]{\smallskip\par\noindent\textsf{\color{black}\textbf{#1}}\hspace{0.5em}\ignorespaces}
\newcommand{\A}{\bA}

\newcommand{\N}{\bN}
\newcommand{\Q}{\bQ}

\newcommand{\B}{\set{\ubegin, \utail}}
\newcommand{\ubegin}{\mathord{\raisebox{-0.3pt}{\scalebox{1.2}{$\bullet$}}}}
\newcommand{\utail}{\mathord{\raisebox{-0.3pt}{\scalebox{1.2}{$\circ$}}}}

\renewcommand{\v}[1]{\boldsymbol{#1}}
\newcommand{\vu}{\v{u}}
\newcommand{\vv}{\v{v}}

\newcommand{\vui}[1]{\v{u}^{(#1)}}
\newcommand{\vvi}[1]{\v{v}^{(#1)}}
\newcommand{\vcu}{\boldsymbol{\check{u}}}
\newcommand{\vcui}[1]{\boldsymbol{\check{u}}^{(#1)}}
\newcommand{\vcup}{\boldsymbol{\check{u}'}}
\newcommand{\vcupp}{\boldsymbol{\check{u}''}}
\newcommand{\vcv}{\boldsymbol{\check{v}}}
\newcommand{\vcvi}[1]{\boldsymbol{\check{v}}^{(#1)}}
\newcommand{\vcvp}{\boldsymbol{\check{v}'}}
\newcommand{\vcvpp}{\boldsymbol{\check{v}''}}

\newcommand{\disjointUnion}{\mathbin{\stackon[-5.2pt]{$\cup$}{$\cdot$}}}

\newcommand{\Subscopes}{\functionT{subs}}

\newcommand{\range}[2]{[{#1}, {#2}]_\N}
\newcommand{\rangeOne}[1]{[{#1}]_\N}

\newcommand{\constr}{\functionT{constr}}
\newcommand{\pad}{\#}
\newcommand{\any}{\functionT{any}}
\newcommand{\pal}{\functionT{pal}}
\newcommand{\Rep}{\mathrm{repr}}
\newcommand{\Id}{\mathrm{Pal}}
\newcommand{\Tr}{\hat{\mathrm{E}}}

\newcommand{\type}{\functionT{type}}
\newcommand{\qfregion}{\functionT{region}}

\newcommand{\limplies}{\rightarrow}

\NewDocumentCommand{\function}{ m e{^_} O{a} m }{%
    \operatorname{%
        #1%
    }%
    \IfValueT{#3}{_{#3}}%
    \IfValueT{#2}{^{#2}}%
    \if\relax\detokenize{#5}\relax
    \else%
        \ifx#4a%
        \mathopen{}\bracketInner*{#5}%
        \else
            \bracketInner[#4]{#5}%
        \fi
    \fi
}

\newcommand{\functionT}[1]{\function{\mathrm{#1}}}
\newcommand{\functionM}{\function}

\newcommand{\X}[1]{\cX(#1)}
\newcommand{\Y}[1]{\cY(#1)}

\newcommand{\Conv}{\functionM{\otimes}}
\newcommand{\model}{\functionM{\mathbb{W}}}
\newcommand{\modelConv}{\functionM{\mathbb{W}}}

\newcommand{\Var}{\cU}
\newcommand{\VarIni}{\cU_\shortuparrow}
\newcommand{\VarOther}{\cU_\shortdownarrow}

\newcommand{\argumentDot}{\mkern2mu\cdot\mkern2mu}
\newcommand{\FOVarOf}{\functionM{\mathrm{var}_1}}

\newcommand{\FreeVarOf}{\functionM{\mathrm{free}}}
\newcommand{\FreeFOVarOf}{\functionM{\mathrm{free}_1\mkern-2mu}} 

\newcommand{\arity}{\functionT{ar}}

\newcommand{\cur}{\smash{\overset{\raisebox{-1pt}[0pt][0pt]{$\scriptstyle\triangledown$}}{y}}}
\newcommand{\valpost}[1]{\smash{\overset{\raisebox{-1pt}[0pt][0pt]{$\scriptstyle\triangleright$}}{#1}}}
\newcommand{\valpre}[1]{\smash{\overset{\raisebox{-1pt}[0pt][0pt]{$\scriptstyle\triangleleft$}}{#1}}}

\newcommand{\Eta}{\mathrm{Eta}}
\newcommand{\Run}{\mathrm{Run}}

\newcommand{\RegOK}{\mathrm{RegOK}}
\newcommand{\Op}{\mathrm{Op}}

\newcommand{\Live}{\mathrm{Live}}
\newcommand{\Pairs}{\mathrm{Pairs}}
\newcommand{\Split}{\mathrm{Split}}

\newcommand{\powerset}{\functionM{\cP}}
\newcommand{\PosOf}{\functionT{pos}}

\newcommand{\lang}{\functionM{\cL}}

\NewDocumentCommand{\parametrised}{ m e{^_} O{} }{%
        \ensuremath{{%
        #1%
    }%
    \IfValueT{#3}{_{#3}}%
    \IfValueT{#2}{^{#2}}%
            \if\relax\detokenize{#4}\relax
            \else
                \squareBracketInner*{#4}\fi}\xspace%
}
\newcommand{\parametrisedT}[1]{\parametrised{\mathrm{#1}}}

\newcommand{\emptyReg}{\textup{nil}}
\newcommand{\emptyRegElement}{\underline{\textup{nil}}}
\newcommand{\Block}{\textup{Block}}

\newcommand{\conf}[1]{C_{#1}}

\newcommand{\ini}{\mathrm{ini}}
\newcommand{\fin}{\mathrm{fin}}

\newcommand{\scope}{S}

\ExplSyntaxOn
\seq_set_from_clist:Nn \l_letter_seq {a,b,c,d,e,f,g,h,i,j,k,l,m,n,o,p,q,r,s,t,u,v,w,x,y,z}
\seq_set_from_clist:Nn \l_capital_letter_seq {A,B,C,D,E,F,G,H,I,J,K,L,M,N,O,P,Q,R,S,T,U,V,W,X,Y,Z}
\seq_map_inline:Nn \l_letter_seq {
    \str_if_eq:nnTF { #1 } { i }
    {
    }
    {
        \cs_new:cpn { f#1 } { \mathfrak{#1} }
    }
    \cs_new:cpn { tt#1 } { \mathtt{#1} }
}
\seq_map_inline:Nn \l_capital_letter_seq {
    \cs_new:cpn { c#1 } { \mathcal{#1} }
    \cs_new:cpn { b#1 } { \mathbb{#1} }
    \cs_new:cpn { tt#1 } { \mathtt{#1} }
    \cs_new:cpn { f#1 } { \mathfrak{#1} }
}
\ExplSyntaxOff
\renewcommand{\cU}{\,\mathcal{U}}

\newcommand{\suchthat}{\suchthatSymbol\PackageWarning{Radeks Macro}{Command suchthat used outside of matching PairedDelimiter was used on input line \the\inputlineno.}}
\newcommand\suchthatSymbol[1][]{\nonscript\:#1\vert\allowbreak\nonscript\:\mathopen{}}
\DeclarePairedDelimiterX{\setInner}[1]\{\}{\renewcommand\suchthat{\suchthatSymbol[\delimsize]}#1}
\NewDocumentCommand{\set}{ O{a} m }{\ifx#1a\setInner*{#2}\else\ifx#1b{\{#2\}}\else\setInner[#1]{#2}\fi\fi}
\NewDocumentCommand{\bracket}{ O{a} m }{\ifx#1a\bracketInner*{#2}\else\bracketInner[#1]{#2}\fi}
\NewDocumentCommand{\squareBracket}{ O{a} m }{\ifx#1a\squareBracketInner*{#2}\else\squareBracketInner[#1]{#2}\fi}

\DeclarePairedDelimiter {\length}        \lvert\rvert
\DeclarePairedDelimiter {\size}          \lvert\rvert
\DeclarePairedDelimiter {\squareBracketInner} []
\DeclarePairedDelimiter {\sem}           \llbracket\rrbracket

\DeclarePairedDelimiter {\bracketInner}  ()

\protected\def\verythinspace{%
    \ifmmode\mskip0.5\thinmuskip\else\ifhmode\kern0.08334em\fi\fi%
}

\let\forallSymbol=\forall
\let\existsSymbol=\exists

\newcommand{\segmentSymbol}{\reflectbox{\tikz[baseline={(char.base)}]{
    \node[anchor=base,inner xsep=-0.22ex,inner ysep=-0.032em] (char) {\phantom{${\existsSymbol}$}};
    \path (char.west) -- (char.north west) coordinate[pos=0.55] (midpoint1);
    \path (char.west) -- (char.south west) coordinate[pos=0.4] (midpoint2);
    \path (char.east) -- (char.north east) coordinate[pos=0.5] (midpoint3);
    \path (char.east) -- (char.south east) coordinate[pos=0.45] (midpoint4);
    \draw[line width=0.065em]
    (midpoint2) to[out=270,in=180] (char.south) to[out=0,in=270] (midpoint4) to[out=90,in=270] (midpoint1) to[out=90,in=180] (char.north) to[out=0,in=90] (midpoint3);
}}}

\RenewDocumentCommand{\forall}{m t.}{%
        {{\forallSymbol}\verythinspace\IfValueT{#1}{#1\IfBooleanT{#2}{.\thinspace}\;}}%
}
\RenewDocumentCommand{\exists}{m t.}{%
        {{\existsSymbol}\verythinspace\IfValueT{#1}{#1\IfBooleanT{#2}{.\thinspace}\;}}%
}
\NewDocumentCommand{\segment}{m t.}{%
        {{\verythinspace\segmentSymbol\verythinspace}_{#1}\IfBooleanT{#2}{.\thinspace}\;}%
}

\makeatletter
\newcommand*{\coloncoloneqqq}{%
    \ensuremath{%
        \mathrel{%
            \@center@colon
            \colonsep
            \@center@colon
            \doublecolonsep
            {\equiv}
        }%
    }%
}
\makeatother

\hideLIPIcs
\nolinenumbers

\title{Scoped MSO, Register Automata, and Expressions: Equivalence over Data Words}

\begin{document}
    \maketitle

    \begin{abstract}
        This paper establishes logical and expression-based characterizations for the class of languages recognized by nondeterministic register automata with guessing (NRA) over infinite alphabets. We introduce Scoped MSO, a logic featuring a novel segment modality and syntactic restrictions on data comparisons. We prove this logic is expressively equivalent to NRA over data domains where ``strong guessing'' can be eliminated. Furthermore, we define Data-Regular Expressions, a minimalist regular-expression calculus built from quantifier-free regions and equipped with $k$-contracting concatenation, and demonstrate its equivalence to NRA over arbitrary relational structures. Together, these formalisms provide a robust descriptive theory for register automata, bridging the gap between automata, logic, and expressions.
    \end{abstract}

    \section{Introduction}
    \label{sec:introduction}
    A remarkable property of the class of \emph{regular languages} is its striking \emph{robustness}: one and the same class of languages admits several radically different, yet equivalent, descriptions. Over a finite alphabet $\Sigma$, finite automata, regular expressions, and monadic second-order logic over word structures (MSO) define exactly the regular languages. This three-way equivalence (and its numerous refinements via algebra, temporal logics, varieties, etc.) forms the backbone of classical language theory and explains why ``regularity'' is such a stable and reusable concept.

When moving from finite alphabets to \emph{infinite alphabets}---modelling systems with process IDs, database keys, large alphabets like Unicode or data values in XML---this robustness generally vanishes. The standard abstraction for such traces is \emph{data words}, sequences over $\Sigma \times \A$ where $\Sigma$ is a finite set of labels and $\A$ is a data domain called \emph{atoms} modelled as a relational structure with countably-infinite universe. The intended access to elements of $\A$ (called \emph{data values}) is limited to testing relations from the signature of $\A$; common examples of atoms include equality atoms $\A_{=} =(\N, =)$ and dense order atoms $\A_{<} = (\Q, <, =)$. The landscape of formalisms for data words is fragmented, featuring a wide array of automata types and logics whose expressive powers are often incomparable and whose decision problems' complexity and decidability vary wildly~\cite{NevenSchwentickVianu04,10.1007/11874683_3,BjorklundSchwentick10,bojanczyk2019slightly,ndet-and-condet-implies-det}. Unlike the finite-alphabet case, there is no single notion that plays the role of ``\emph{the}'' regular languages for data words.

\namedparagraph{Register Automata.}
A natural candidate for a finite-state model over data words is the \emph{register automaton}, introduced by Kaminski and Francez~\cite{Kaminski1990FinitememoryA}.
These automata extend finite control with a fixed set of registers capable of storing data values and comparing them for equality.
The model was subsequently extended to arbitrary atoms $\A$ \cite{BojanczykKlinLasotaNominalAutomata}.
In this paper, we target the class of \emph{nondeterministic register automata with guessing} (\NRAs). The ``guessing'' capability is the most permissive out of three possible variants of nondeterminism in register automata. While \NRAs are a standard baseline, they have historically lacked a MSO-style logical characterisation.

\namedparagraph{The MSO Logic for NRAs.}
Establishing a correspondence between \NRAs and logic requires careful calibration of the expressive power.
Simply adding a data atomic formulas to \MSOnoalpha leads to undecidability~\cite{NevenSchwentickVianu04,10.1145/1970398.1970403}, even for the first-order (FO) fragment with three variables and the simplest variant of atoms $\A_{=}$: one may encode grid-like structures of arbitrary dimension in data words and simulate runs of Turing machines on them.
Conversely, restricting the logic too aggressively (e.g., to rigid guards, or to the two-variable FO fragment) yields decidability but fails to capture the full expressive power of \NRAs~\cite{ColcombetLeyPuppis15,BojanczykStefanski20,10.1145/1970398.1970403}.
The lack of closure of \NRAs under language complement~\cite{ndet-and-condet-implies-det} poses yet another challenge---any logic of the same expressive power must impose limitations on negation operation.

\namedparagraph{Expressions for NRAs.}
Regarding expression formalisms, Brunet and Silva define expressions featuring explicit binders for name allocation and deallocation~\cite{BrunetSilva19} and prove an analogue of the Kleene theorem: that these expressions define exactly languages of \NRAs.
Their formalism provides a versatile \emph{specification} language, allowing concise and readable expressions at the cost of more involved definition of semantics.
We aim for a complementary goal: a \emph{minimalist} expression formalism whose syntax is close in spirit to classical regular expressions.

\namedparagraph{Contributions.}
In this paper, we provide two new formalisms for languages of data words over $\Sigma \times \A$: Scoped MSO logic (\SMSO[\Sigma, \A]) and Data-Regular Expressions (\DRE[\Sigma,\A]), recovering the classical trinity for NRAs:
\begin{align*}
    \text{Automata ($\NRA[\Sigma, \A]$)} \;\equiv\; \text{Expressions ($\DRE[\Sigma, \A]$)} \;\equiv\; \text{Logic (\SMSO[\Sigma,\A])}\,.
\end{align*}

\begin{enumerate}
    \item \textbf{Data-Regular Expressions:} We introduce a compact expression formalism over $\Sigma \times \A$.
    \DREs are built from \emph{quantifier-free regions} of data words.
    The calculus relies on a parameterized operation on languages called \emph{$k$-contracting concatenation}. The parameter $k$ exposes a bounded ``interface'' between factors, reflecting the finite register set of the automaton.

    \smallskip
        {\color{lipicsGray}$\blacktriangleright$} We prove that $\lang{\DRE[\Sigma,\A]} = \lang{\NRA[\Sigma,\A]}$.\smallskip

    \item \textbf{Scoped MSO:} We define a logic tailored to the expressive power of \NRAs.
    \SMSOnoalpha extends MSO with two features:
    \begin{itemize}
        \item A \emph{segment modality} $\segment{X} \varphi$, which evaluates $\varphi$ on sub-intervals defined by a set of cut-positions $X$.
        \item \emph{Restricted atomic predicates} $R(\v x)$ with a \emph{syntactic restriction} preventing arbitrary data comparisons that would exceed the capacity of finite registers.
    \end{itemize}
        {\color{lipicsGray}$\blacktriangleright$} We prove that $\lang{\SMSO[\Sigma,\A]} = \lang{\NRA[\Sigma,\A]}$, provided the atoms $\A$ admit \emph{elimination of strong guessing}---a technical property we formalize.
    This condition holds for standard domains, including equality atoms $\A_{=}$ and dense order atoms $\A_{<}$, and implies decidability of satisfiability for $\SMSOnoalpha$ over these atoms.
    We discuss natural adaptations of $\SMSOnoalpha$ for $\omega$-words and subclasses of \NRAs.
\end{enumerate}

\namedparagraph{Implications and Outlook.}
These results provide a descriptive theory for NRAs that parallels the finite-alphabet picture.
Furthermore, this new characterisation offers a toolkit for attacking open problems in the theory of data languages.
For instance, the logical formalism of $\SMSOnoalpha$ may provide the missing insights to resolve Colcombet's conjecture regarding the separability of disjoint \NRA languages by Unambiguous Register Automata~\cite{10.1007/978-3-319-19225-3_1}\footnote{In the cited paper \cite{10.1007/978-3-319-19225-3_1} stated as a theorem, later regained the status of a conjecture.}.
Additionally, we leave open whether there is any strengthening
of $\SMSO[\Sigma,\A]$ which retains decidable satisfiability, while going beyond the class of \NRA languages.

    \section{Preliminaries}
    \label{sec:preliminaries}
        We use standard notations to denote natural numbers $\N$, positive natural numbers $\N_+$ or rationals $\Q$.
    We denote ranges of natural numbers by $\range{i}{j} \coloneqq \N \cap [i,j]$ and $\rangeOne{j} \coloneqq \range{1}{j}$.

    \namedparagraph{Words.}
    Fix an alphabet $\Sigma$ and a word $w = a_1 a_2 \cdots a_n \in \Sigma^*$.
    We write $w[i] \coloneqq a_i$ and, for an interval $I = \range{i}{j}$, we write $w[I]$ for the infix $a_i a_{i+1} \cdots a_j$.
    The reversed word $a_n \cdots a_1$ is denoted as $w^R$.
    The length $n$ of $w$ is denoted as $\length{w}$, and the set of \emph{positions} of $w$ is $\PosOf{w} = \rangeOne{n}$.
    The unique word of length $0$ is called an \emph{empty word} and denoted by $\varepsilon$.
    We identify the set of families of words $(Y^*)^X$ with the function space $X \to Y^*$.
    Accordingly, we introduce families as mappings $\v u\colon X \to Y^*$ but retain the subscript notation $u_x$ for components.
    We lift concatenation and infix extraction operations to families of words.

    \namedparagraph{Word structures.}
        \label{def:word-structure-finite-alphabet}
        Given a $w \in \Sigma^*$, a \emph{word structure}
        $\model{w}$ is the relational structure defined as
        $\model{w} = (\PosOf{w}, <, ({\sigma})_{\sigma \in \Sigma})$ where ``$<$'' is interpreted as the natural linear ordering on $\PosOf{w}$,
        and for each letter $\sigma \in \Sigma$ we have a unary relation symbol $\sigma = \set{ p \in \PosOf{w} \suchthat w[p] = \sigma }$ interpreted the set of positions in $w$ labelled with $\sigma$.
        While we use $\sigma$ both for letters of $\Sigma$ and position predicates, the intended meaning is clear from the context.

    \subsection{Formalisms for regular languages}
    \label{subsec:formalisms-for-regular-languages}
    This paper's starting point is the equivalence of the following three formalisms defining regular languages.
    Fix a finite alphabet $\Sigma$.

    \namedparagraph{Nondeterministic finite automata.}
    A \emph{Nondeterministic Finite Automaton} (NFA) is a tuple $\cA = (Q, \Sigma, Q_\ini, Q_\fin, \delta)$, where $Q$ is a finite set of states, $Q_\ini$ and $Q_\fin$ its initial and accepting subsets, and $\delta \subseteq Q \times \Sigma \times Q$.
    We call the elements of $\delta$ \emph{transitions} and denote $(q, a, r) \in \delta$ using the arrow notation $q \xrightarrow{a} r$.
    A \emph{run} of $\cA$ of length $n$ over $w = a_1 a_2 \cdots a_n$ is a $w$-labelled path $\pi = t_1 t_2 \cdots t_n \in \delta^*$ in the labelled graph $(Q, \delta)$, i.e., a sequence such that $t_i$ has the form $(\argumentDot, a_i, \argumentDot)$ for every $i$, and for any two consecutive transitions $p \xrightarrow{a} q$ and $r \xrightarrow{b} s$ we have $q = r$.
    We say that $\cA$ \emph{accepts} $w \in \Sigma^*$ if there exists a run of $\cA$ over $w$ that begins in some initial state $q_\ini \in Q_\ini$ and ends in an accepting state $q_\fin \in Q_\fin$.
    The language of $\cA$ is $\lang{\cA} = \set{w \in \Sigma^* \suchthat \text{$\cA$ accepts $w$}}$.
    We write $\NFA[\Sigma]$ for the family of all \NFAs~over~$\Sigma$.

    \namedparagraph{Regular expressions.}
    \emph{Regular expressions} over $\Sigma$ ($\RE[\Sigma]$) are terms generated by the grammar
    \begin{align*}
        E, F \coloncoloneqqq \emptyset \mid \epsilon \mid a \in \Sigma \mid E + F \mid E \cdot F \mid E^*\,.
    \end{align*}
    The language $\lang{E}$ of an expression $E$ is defined as follows
    \begin{align*}
        \lang{\emptyset} &= \emptyset &
        \lang{\epsilon} &= \set{\epsilon} &
        \lang{a} &= \set{a} \\
        \lang{E + F} &= \lang{E} \cup \lang{F} &
        \lang{E \cdot F} &= \lang{E} \cdot \lang{F} &
        \lang{E^*} &= \textstyle\bigcup_{n \ge 0} \lang{E}^n
    \end{align*}
    where the language concatenation is $K\cdot L \coloneqq \set{uv\suchthat u \in K, v \in L}$ for $K, L \subseteq \Sigma^*$, and $K^n$ denotes $n$-fold concatenation of $K$ with itself; $K^0 \coloneqq \set{\varepsilon}$.

    \begin{theorem}[Kleene \cite{kleene1956representation}]
        \label{thm:kleene-regular}
        $\lang{\NFA[\Sigma]} = \lang{\RE[\Sigma]}$ for any finite $\Sigma$.
    \end{theorem}

    \namedparagraph{Syntax of MSO.}
        \label{def:mso-logic}
        The syntax of \emph{Monadic Second Order Logic over Word Structures}~(\makebox{$\MSO$}) features first-order variables $x, y, \ldots$, second-order variables $X, Y, \ldots$ and is given by the grammar
        {{\newcommand{\entry}[1]{#1}
    \newcommand{\entryb}[1]{\hspace{0.21cm}\clap{$#1$}\hspace{0.21cm}}
    \begin{align*}
        \varphi, \psi \coloncoloneqqq {}
        &\entry{\exists{x} \varphi} \mid \entry{\exists{X} \varphi} \mid
        \entry{\varphi \land \psi} \mid \entry{\varphi \lor \psi} \mid \entry{\neg \varphi} \mid
        \entry{x < y} \mid \entry{X(x)} \mid \entry{a(x)}\,,
    \end{align*}}
        where $a \in \Sigma$.
        We write $\varphi \in \MSO$ when $\varphi$ is generated by the above grammar.
        A variable in $\varphi$ is free, if it is not bound by any quantifier.
        We write $\FreeVarOf{\varphi}$ (or $\FreeFOVarOf{\varphi}$) for the set of free variables (of first-order variables, respectively) of $\varphi$.
        A formula is a \emph{sentence}, if it has no free variables.

    \namedparagraph{Semantics of MSO.}
        The truth value of $\MSO$ formulas depends on the interpretation $\fI = (\model{w}, \nu)$ consisting of the word structure $\bW = \model{w}$ and valuation $\nu \colon \FreeVarOf{\varphi} \to \PosOf{w} \cup \powerset{\PosOf{w}}$.
        Boolean connectives and atomic formulas ($x<y, X(y), a(x)$) are interpreted in a standard way over $\bW$.
        In particular, $\fI \models a(x)$ iff $\nu(x) \in a^\bW$.
        Quantifiers range over $\PosOf{w}$:
        \begin{itemize}
            \item \makebox[16mm][l]{$\fI \models \exists{x} \varphi$} iff there is a position $p \in \PosOf{w}$ such that $\model{w}, \nu[x \mapsto p] \models \varphi$,
            \item \makebox[16mm][l]{$\fI \models \exists{X} \varphi$} iff there is a set $P \subseteq \PosOf{w}$ such that $\model{w}, \nu[X \mapsto P] \models \varphi$.
        \end{itemize}
        For a sentence $\varphi \in \MSO[\Sigma]$, we write $\model{w} \models \varphi$ whenever $\model{w}, \nu_0 \models \varphi$, where $\nu_0$ is the empty valuation.
        The language of $\varphi$ is defined as $\lang{\varphi} = \set{w \in \Sigma^* \suchthat \model{w}\models \varphi }$.

    \begin{theorem}[B\"uchi, Elgot, Trakhtenbrot \cite{Buchi-mso-automata,elgot1961decision,zbMATH03186871}]
        \label{thm:buchi-elgot-trakhtenbrot}
        For every finite $\Sigma$, the languages definable by $\MSO[\Sigma]$ are exactly the regular languages over $\Sigma$.
    \end{theorem}

    \namedparagraph{Syntactic sugar.}
    We use $\forall{\xi} \varphi$ as a shortcut for $\neg \exists{\xi} (\neg \varphi)$; position equality is defined as $x=y \equiv \neg(x < y \lor y < x)$; Boolean constants can be defined as $\top \equiv \exists{x} x = x \lor \neg \exists{x} x = x$ and $\bot \equiv \neg \top$.
    Additionally, we add easily MSO-definable relations like subset relation $S' \subseteq S$.

    \subsection{Words over infinite alphabets}
    While the definitions of a word and a language do not require \emph{finiteness of the alphabet}, all the basic language-defining formalisms we have recalled in \cref{subsec:formalisms-for-regular-languages} critically depend on that assumption.
    For instance, a finite automaton for the language $\Sigma$ consisting of one-letter words has to contain a transition for every $a \in \Sigma$, and making $\Sigma$ infinite would cause the transition relation to become infinite as well.
    One solution to this problem is to add more structure to the alphabet and allow the model to make less precise queries about the letters.
    This gives rise to the realm of \emph{atoms} and \emph{register automata}.
    \namedparagraph{Atoms.} We call atoms any relational structure $\A = (U, R_1, \ldots, R_\ell)$ with a countably infinite universe $U$ and a finite number of relational symbols in its signature.
    We assume that that one of $R_i$ is the equality of universe elements.
    We deliberately confuse $\A$ and its universe $U$ and write $\alpha \in \A$.
    \namedparagraph{Data words.} Data words are words over a \emph{two-track} alphabet $\Sigma \times \A$, consisting of a finite label from $\Sigma$ and a data value from $\A$.
    It is often convenient to decompose multi-track words into separate tracks using \emph{convolution}.
    For alphabets $A_1, \dots, A_k$, the convolution of words $w_1 \in A_1^*, \dots, w_k \in A_k^*$ of equal length is the word $\Conv{w_1, \dots, w_k} \in (A_1 \times \cdots \times A_k)^*$ defined by $\Conv{w_1, \dots, w_k}[i] \coloneqq (w_1[i], \dots, w_k[i])$.
    Thus, a data word can be viewed as the convolution of a label word over $\Sigma$ and a data sequence over $\A$.

    \begin{example} Let $\A = (\N, =)$ and $\Sigma = \set{a, b}$.
    A two-track word $\Conv{w, \varpi}$ over $\Sigma \times \A$ is a convolution of two tracks $w \in \Sigma^*$ and $\varpi \in \A^*$:
    \begin{align*}
        &\begin{aligned}
             w &= \text{abaab} \\[-1mm]
             \varpi &= 10, 42, 5, 97, 42
        \end{aligned}
        &
        w &=
        \begin{bmatrix}
            \text{a} \\
            10
        \end{bmatrix}
        \,
        \begin{bmatrix}
            \text{b} \\
            42
        \end{bmatrix}
        \,
        \begin{bmatrix}
            \text{a} \\
            \,5\,
        \end{bmatrix}
        \,
        \begin{bmatrix}
            \text{a} \\
            97
        \end{bmatrix}
        \,
        \begin{bmatrix}
            \text{b} \\
            42
        \end{bmatrix}
    \end{align*}
    Above, $\varpi[2] = \varpi[5] = 42$ and $w[3]$ is $a$.
    \end{example}

    \namedparagraph{Quantifier-free formulas and regions over $\bm{\A}$.}
    Let $\cX$ be a finite set of variables.
    The set of \emph{atomic formulas} over $\A$ using variables $\cX$ consists of all expressions $R(x_1, \dots, x_r)$ where $R$ is a relation symbol from the signature of $\A$ and $x_i \in \cX$.
    A \emph{quantifier-free formula} over $\A$ using variables $\cX$ is a Boolean combination of atomic formulas:
    \begin{align*}
        \varPhi, \Psi \coloncoloneqqq \top \mid \neg \varPhi \mid \varPhi \land \Psi \mid R(x_1, \dots, x_r)\,.
    \end{align*}
    The formulas are interpreted in a standard way over a pair $(\A, \nu)$, where $\nu\colon \cX \to \A$.

    For a word $w = a_1 \cdots a_k \in \A^k$, the \emph{quantifier-free type} of $w$, denoted $\type{w}$, is the set of all literals (atomic formulas or their negations) using variables $\cX_k \coloneqq \set{x_1, \ldots, x_k}$ that are satisfied when variables are interpreted as data values in $w$. Formally, $\type{w} \coloneqq \set{ \phi \in \text{Lit}(\cX_k) \suchthat \A, \nu \models \phi }$ where $\nu(x_i) = a_i$.
    A \emph{quantifier-free region} of a word $w \in \A^*$, denoted $\qfregion_\A{w}$, is the set of all words that have the same quantifier-free type as $w$.
    Intuitively, words in $\qfregion_\A{w}$ are not distinguishable from $w$ by quantifier free formulas.
    We lift $\qfregion_\A{}$ to data words in a natural way.
    It is easy to see that $\A^k$ is a finite union of quantifier-free regions for any $k$.

    \namedparagraph{Data word structures.}
        \label{def:data-word-structure}
        Fix atoms $\A=(U,R_1^\A,\ldots,R_\ell^\A)$ and let $\Conv{w,\varpi}\in(\Sigma\times U)^*$ with
        $\length{w}=\length{\varpi}$.
        The \emph{data word structure} is the expansion of $\model{w} = (\PosOf{w},<,(P_a)_{a\in\Sigma})$
        \[
            \modelConv{w,\varpi} \coloneqq (\PosOf{w},<,(\sigma)_{\sigma\in\Sigma},(R_i)_{i\in\rangeOne{\ell}})
        \]
        where, if $R_i$ has arity $m_i$, then for all $p_1,\ldots,p_{m_i}\in\PosOf{w}$,
        \[
            R_i(p_1,\ldots,p_{m_i})
            \quad\text{iff}\quad
            R_i^{\A}(\varpi[p_1],\ldots,\varpi[p_{m_i}]).
        \]

    \subsection{Register automata---a formalism for data word languages}
    \label{subsec:a-formalism-for-data-word-languages}
    One of the common formalisms for defining languages of data words is a nondeterministic register automaton (\NRA).
    Fix atoms $\A = (U, R_1, \ldots, R_\ell)$ and a finite $\cR$ whose elements we call \emph{registers}.
    Let $k \coloneqq \size{\cR}$ and let $\A_\emptyReg \coloneqq \A \disjointUnion \set{\emptyRegElement}$, where $\emptyReg$ is a fresh element modelling an empty register, and the signature of $\A_\emptyReg$ is extended with a unary symbol $\emptyReg$ interpreted as $\set{\emptyRegElement}$.
    A register valuation is a function $\mu\colon \cR \to \A_\emptyReg$.

    \namedparagraph{Register constraints.}
    Let $\cV_\cR \coloneqq \bigcup_{r \in \cR} \set[\big]{\valpre{r}, \valpost{r}}$ be the variables representing register values before and after a transition, and let $\cur$ represent the current input symbol.
    A \emph{register constraint} $\varPhi$ is a quantifier-free formula over $\A_\emptyReg$ with the variable set $\set{\cur} \cup \cV_\cR$.
    We denote the set of all register constraints for $\cR$ as $\constr_\A{\cR}$.
    To evaluate a constraint, consider two valuations $\mu, \mu' \colon \cR \to \A \cup \{\emptyReg\}$ and an input symbol $\alpha \in \A$.
    These induce a combined valuation $\nu$ on the variables $\set{\cur} \cup \cV_\cR$ defined by $\nu(\valpre{r}) = \mu(r)$, $\nu(\valpost{r}) = \mu'(r)$, and $\nu(\cur) = \alpha$.

    \namedparagraph{Register automata.}
    A \emph{nondeterministic register automaton with guessing} (\NRA) is a tuple $\cA = (Q, \Sigma, \A, \cR, Q_\ini, Q_\fin, \delta)$, comprising a finite set of states $Q$, a finite alphabet $\Sigma$, initial states $Q_\ini \subset Q$, accepting states $Q_\fin \subseteq Q$ and a set of \emph{transition rules} $\delta\colon Q \times \Sigma \times \constr_\A{\cR} \times Q$.
    A transition rule $(p, \sigma, \varPhi, q) \in \delta$ is written using double arrow notation $p \xRightarrow{\sigma, \varPhi} q$.
    The semantics is defined by an infinite transition system $\sem{\cA} \coloneqq (\conf{\cA}, \Delta_\cA)$ over configurations $\conf{\cA} \coloneqq Q \times (\cR \to \A_\emptyReg)$, where $\Delta \subseteq \conf{\cA} \times \Sigma \times \A \times \conf{\cA}$.
    A transition $(p, \mu) \xrightarrow{\sigma, \alpha} (q, \mu') \in \Delta_\cA$ exists iff there is a rule $p \xRightarrow{\sigma, \varphi} q$ in $\delta$ such that $\A, \nu(\mu, \alpha, \mu') \models \varPhi$.
    A configuration $(q, \mu)$ is initial if $q \in Q_\ini$ and $\mu(r) = \emptyReg$ for every $r \in \cR$.
    A configuration $(q, \mu)$ is \emph{accepting} when $q \in Q_\fin$.
    The notion of a run is analogous as in \NFAs; it is accepting if it starts in some initial configuration and ends in some accepting one.
    The language $\lang{\cA}$ consists of data words labelling runs of $\sem{\cA}$.
    Every run $\pi$ of $\sem{\cA}$ has an associated sequence $t_1\cdots t_n \in \delta^*$ of transition rules.
    We denote the family of $k$-register NRA over $\Sigma \times \A$ by $\NRA_k[\Sigma, \A]$.

    \namedparagraph{Variants of nondeterminism in register automata.}
    Fix $\cA \in \NRA_k$ and its run $\pi = (q_0, \mu_0) \xrightarrow{\sigma_1,\alpha_1} (q_1, \mu_1) \rightarrow\cdots\rightarrow (q_{n-1}, \mu_{n-1}) \xrightarrow{\sigma_n, \alpha_n} (q_n, \mu_{n})$ over $w \in (A \times \Sigma)^n$.
    We say that $\cA$ is \emph{guessing} the value of $j$th register in valuation $\mu_a$, if $\xi = \mu_{a}(r_{j})$ is different from all previous register values $\mu_{a-1}(\cR)$ and from the data value $\alpha_a$.
    The guess can be either \emph{weak} or \emph{strong}, depending on the run from configuration $(q_a, \mu_a)$th onwards.
    We say that the \emph{live interval} of $\xi$ is $\range{a}{b}$ whenever $\xi \in \mu_m(\cR)$ for all $m \in \range{a}{b}$ and $\xi \notin \mu_{b+1}(\cR)$.
    The guess of the value of $j$th register in valuation $\mu_a$ is strong whenever $\alpha_m \neq \xi$ for every $m$ in the live interval $\range{i}{j}$.
    
    We distinguish two restricted subclasses of \NRA:
    \emph{weakly-guessing nondeterministic register automata} (\NRAweak) and \emph{deterministic register automata without guessing} (\NRAnoguess), which arise as syntactic limitations on the transition constraints $\varPhi$ and the determinism of $\delta$.
    We say that $\cA \in \NRA[\Sigma, \A]$ is weakly guessing, if $\cA$ is not guessing strongly in its accepting runs.
    It is without guessing, if it is never guessing a value in any of its runs.
    
    \begin{definition}[Elimination of strong guessing.]
        \label{def:elimination-of-strong-guessing}
        We say that atoms $\A$ have elimination of strong guessing, whenever for every $\cA \in \NRA[\Sigma, \A]$ there exists a weakly guessing $\cB \in \NRA[\Sigma, \A]$ that recognises the same language.
    \end{definition}

    \section{New formalisms equivalent to register automata}
    \label{sec:new-formalisms-equivalent-to-nra}
    \subsection{Data regular expressions}
\label{sec:contribution}

    Fix atoms $\A$ and a finite $\Sigma$.
    We introduce a notion of \emph{data regular expressions} over $\Sigma \times \A$, which are generated by the following grammar
    \begin{align*}
        E, F \coloncoloneqqq \emptyset \mid [w] \mid E + F \mid E \cdot_k F \mid E^{k,*}\,,
    \end{align*}
    where $k \in \N$ and $w \in (\Sigma \times \A)^*$.
    Let $\DRE[\Sigma, \A]$ be the set of all such data expressions.
    The language $\lang{E}$ of $E \in \DRE[\Sigma, \A]$ is defined as follows:
    \begin{align*}
        \lang{\emptyset} &= \emptyset &
        \lang{[w]} &= \qfregion_\A{w} \\
        \lang{E + F} &= \lang{E} \cup \lang{F} &
        \lang{E \cdot_k F} &= \lang{E} \cdot_k \lang{F} &
        \lang{E^{k,*}} &= \textstyle\bigcup_{n \geq 0} \lang{E}^{k,n}
    \end{align*}
    where $\cdot_k$ is \emph{$k$-contracting concatenation} defined as $K\cdot_k L \coloneqq \set{uw\suchthat uv \in K, v^Rw \in L, \length{v} = k}$ for data word languages $K, L$, and $K^{n,k}$ for $n > 0$ denotes the $n$-fold $k$-contracting concatenation of $K$ with itself, whereas for $n=0$ we define $K^{0,k} \coloneqq \set{vv^R \suchthat \length{v} = k}$, the set of palindromes of length $2k$.
\smallskip
\begin{theorem}
    \label{thm:expressions-NRA-equivalence}
    $\lang{\NRA[\Sigma, \A]} = \lang{\DRE[\Sigma, \A]}$\marginpar{\hfill$\star$}
    for any finite $\Sigma$ and atoms $\A$.
\end{theorem}

\subsection{Scoped MSO}
\namedparagraph{Syntax of Scoped MSO.}
    Fix finite $\Sigma$ and atoms $\A = (U, R_1, \ldots, R_\ell)$.
    The syntax of Scoped MSO uses first-order variables $x, y, \ldots$, second-order variables $X, Y, \ldots$, and is generated by
    \newcommand{\entry}[1]{#1}
    \begin{align*}
        \varphi, \psi \coloncoloneqqq {}
        &\entry{\exists{x} \varphi} \mid \entry{\exists{X} \varphi} \mid \entry{\segment{X} \varphi} \mid
        \entry{\varphi \land \psi} \mid \entry{\varphi \lor \psi} \mid \entry{\neg \varphi} \mid
        \entry{x < y} \mid \entry{X(x)} \mid \entry{a(x)} \mid \entry{R(\v x)}\,,
    \end{align*}
    for all letters $a \in \Sigma$, relations $R \in \set{R_1, \ldots, R_\ell}$ and where $\v{x} = (x_1, \ldots, x_r)$ matches the arity of $R$.
    There are two additions compared to the grammar of MSO (as per \cref{def:mso-logic}): $\segmentSymbol_X$ is the \emph{scope modality} and $R(\v x)$ is a \emph{data atomic formula}.

A first-order variable in $\varphi$ is \emph{top-level} if its binding existential quantifier does not occur within the scope of a negation; otherwise, it is \emph{nested}.
(Note that variables in $\forall x \varphi \equiv \neg \exists x \neg \varphi$ are inherently nested).

\begin{condition}[Well-formedness]
    \label{cond:SMSO-condition}
    A formula $\varphi$ is \emph{well-formed} if every data atomic formula $R(\v x)$ occurring in $\varphi$ contains at most one nested variable.
\end{condition}
\begin{example}
    Fix $\A = (\N, \sim)$, where $\sim$ is the equality relation on $\N$.
    The formula $\exists{x} \neg [\exists{y}. x \sim y \land x \neq y]$ is well-formed because $x$ is top-level.
    However, $\neg [\exists{x} \exists{y}. x \sim y \land x \neq y]$ is not, as the outer negation makes both $x$ and $y$ nested.
    The modality $\segmentSymbol_X$ preserves the top-level status of variables; thus, $\exists{X} \segment{X} \exists{x} \neg [\exists{y}. x \sim y \land x \neq y]$ remains well-formed.
\end{example}
We denote the set of all well-formed formulas as $\SMSO$.

\namedparagraph{Subscopes.}
    Let $I = \range{a}{b}$ and let $X \subseteq \N$.
    We view $X$ as a set of \emph{split points}, and we define the set $\Subscopes{I, X}$ of the contiguous
    sub-ranges obtained by cutting $I$ at every position in $X$.
    Formally, let $\Subscopes{I, X}$ be the set of all $\range{a'}{b'}$ such that $a', b'+1$ are two consecutive elements of the sorted set
    $(X \cap I) \cup \{a,\, b+1\}$.
\begin{example}
    Let $I = \range{a}{b}= \range{10}{20}$ and $X = \set{5, 7, 11, 13, 17, 23}$.
    The set of split-points $(X \cap I) \cup \{a, b+1\}$ is $\set{10, 11, 13, 17, 21}$, and therefore
    \begin{align*}
        \Subscopes{I, X} =
        \{\underset{\mathclap{\stackrel{\kern1pt\rotatebox{90}{{\phantom{$\in$}\llap{$=$}}}}{\phantom{X}a\phantom{X}}}}{10}\},
        \{\underset{\mathclap{\stackrel{\rotatebox{90}{{$\ni$}}}{X}}}{11}, 12\},
        \{\underset{\mathclap{\stackrel{\rotatebox{90}{{$\ni$}}}{X}}}{13}, 14, 15, 16\},
        \{\underset{\mathclap{\stackrel{\rotatebox{90}{{$\ni$}}}{X}}}{17}, 18, 19, \underset{\mathclap{\stackrel{\kern1pt\rotatebox{90}{{\phantom{$\in$}\llap{$=$}}}}{b+1}}}{20}\}\,.
    \end{align*}
\end{example}

\namedparagraph{Semantics of Scoped MSO.}
    We interpret $\SMSO$ formulas over a tuple $\fI = (\modelConv{w, \varpi}, \nu, I)$ consisting of the data word structure $\bW = \modelConv{w, \varpi}$, valuation $\nu \colon \FreeVarOf{\varphi} \to \PosOf{w} \disjointUnion \powerset{\PosOf{w}}$ and a discrete range $I = \range{i}{j} \subseteq \PosOf{w}$ called a \emph{scope}.
    Boolean connectives and atomic formulas ($x<y, X(y), a(x), R(\v x)$) are interpreted in a standard way over $(\bW, \nu)$, ignoring $I$.
    In particular, $\fI \models R(\v x)$ iff $\nu(\v x) \in R^{\bW}$.
    Quantifiers range strictly over the current scope $I$:
    \begin{itemize}
        \item \makebox[16mm][l]{$\fI \models \exists{x} \varphi$} iff \kern2pt there is a position $p \in I$ such that $\modelConv{w, \varpi}, \nu[x \gets p], I \models \varphi$,
        \item \makebox[16mm][l]{$\fI \models \exists{X} \varphi$} iff \kern2pt there is a set of positions $P \subseteq I$ such that $\modelConv{w, \varpi}, \nu[X \gets P], I \models \varphi$.
    \end{itemize}
    Finally, $\segmentSymbol_X$ enforces the formula on all induced sub-scopes:
    \begin{itemize}
        \item \makebox[16mm][l]{$\fI \models \segment{X} \varphi$} iff \kern2pt $\modelConv{w, \varpi}, \nu, I' \models \varphi$ for every $I' \in \Subscopes{I, \nu(X)}$.
    \end{itemize}
    When the scope interval is the entire $\PosOf{w}$, we drop it in the interpretation tuple and write $\modelConv{w, \varpi}, \nu \models \varphi$.
    When $\varphi$ is a sentence, we drop the valuation and write $\modelConv{w, \varpi} \models \varphi$, and
    we define the \emph{language} of $\varphi$ as $\lang{\varphi} = \set{ \Conv{w, \varpi} \suchthat \modelConv{w, \varpi} \models \varphi }$.
\vspace{1mm}
\begin{theorem}
    \label{thm:logic-NRA-equivalence} $\lang{\NRA[\Sigma, \A]} = \lang{\SMSO}$ for \marginpar{$\star$}any finite $\Sigma$ and atoms $\A$ that have elimination of strong guessing (as per \cref{def:elimination-of-strong-guessing}).
\end{theorem}
\begin{corollary}
    If atoms $\A$ have elimination of strong guessing, then
    satisfiability is decidable for $\SMSO$ iff $\NRA[\Sigma,\A]$ emptiness is decidable.
\end{corollary}

\begin{lemma}
    \label{lem:elimination-of-strong-guessing}
    Equality atoms $\A_{=} = (\bN, =)$ and dense order atoms $\A_{<} = (\bQ, <)$ both have elimination of strong guessing.
\end{lemma}
Section \cref{sec:extensions} discusses how assumption about strong guessing elimination can be lifted at the cost of making the logic more convoluted.

    \section{Proofs of main theorems}
    \label{sec:proofs}
    We briefly outline the core ideas behind \cref{thm:logic-NRA-equivalence,thm:expressions-NRA-equivalence}, while
their detailed proofs were relegated to the appendix, \cref{sec:expressions--automata-equivalence,sec:logic--automata-equivalence}.

\subsection{Expressions vs.\ automata}
A $k$-register automaton can only transport $k$ data values across a cut in the
input. The operator $\cdot_k$ of \DREs is designed to reflect exactly this:
a word belongs to $K\cdot_k L$ iff it can be split as $uw$ such that there
exists an overlap $v\in(\Sigma\times\A)^k$ with $uv\in K$ and $v^R w\in L$.
Intuitively, $v$ is a length-$k$ \emph{interface} that transfers the bounded
number of data values needed to continue recognition on the right.
We provide reductions between $\DRE$ and $\NRAs$, as the formalism of expressions with binders~\cite{BrunetSilva19} proved
conceptually difficult to be simulated using $k$-contracting concatenation.

\namedparagraph{From \DREs to \NRAs.}
The compilation is a structural induction on expressions.  The only
non-classical step is $\cdot_k$ (and the corresponding $({k,*})$-iteration):
to realise the existentially-quantified overlap, the automaton
\emph{guesses} the interface word $v$ nondeterministically, stores its $k$ data
values in registers, and then \emph{replays} the overlap (in reverse) when
starting the right component.

\namedparagraph{From \NRAs to \DREs.}
Here the main idea is to separate \emph{finite control} from \emph{data
consistency}.  We first forget data and view a run as a word over the finite
alphabet of transition rules; the set of control-consistent rule sequences is
regular, hence described by an ordinary regular expression (via Kleene).
We then re-inject data by replacing each transition rule by a local
\emph{block language} of length $2k{+}1$ that encodes, in its data track, the
pre-valuation of the $k$ registers, the current datum, and the post-valuation,
and checks the transition constraint as a quantifier-free condition on this
tuple.

\subsection{Logic vs.\ automata}
The equivalence $\lang{\SMSOnoalpha}=\lang{\NRA}$ has one central obstacle: if a logic
can compare data values at unrelated positions freely, satisfiability quickly
becomes undecidable, and in any case it exceeds the power of finitely many
registers.
Our two ingredients address exactly this.
First, well-formedness restricts each data atom to contain at most one nested
(first-order) variable, preventing arbitrary cross-position dependencies under
negation.
Second, the modality $\segment{X}$ provides controlled recursion over
sub-intervals, letting us enforce properties in scopes loosely corresponding to live intervals of registers.

\namedparagraph{From \SMSO\ to \NRAs.}
The idea is to compile away the scoping features, and extract the data value dependencies between positions until only \emph{local} data
tests remain.
Top-level variables are handled by introducing auxiliary tracks that carry, on
each scope, a single chosen witness value propagated across the scope.
This turns nonlocal references (``use the value chosen for $x$ inside the
current scope'') into local track access.

Next, $\segment{X}\varphi$ is unfolded into a quantification over MSO-definable
subscopes and a conjunction of copies of $\varphi$ on those subscopes; after
this step, data predicates in the translated formula become \emph{locally
testing}, i.e.\ they only inspect the current position in the (expanded) word.
Locally testing formulas can be recognised by a $0$-register NRA: each position
has only finitely many quantifier-free types over the data signature, so we can
reduce to a finite-alphabet MSO formula and then to a finite automaton, and
finally re-interpret the types back as quantifier-free constraints on the
current datum.
The auxiliary data value tracks are eliminated by nondeterministic projection: the NRA
guesses them in the registers online and checks the local constraints, additionally
ensuring that the ``witness track'' is constant along the intended segments.
This yields an \NRA over the original alphabet $\Sigma\times\A$.

\namedparagraph{From weakly guessing \NRAs to \SMSOnoalphabold.}
In the reverse direction we assume weak guessing (w.l.o.g.\ under elimination
of strong guessing) and first normalise to a simple form where register updates
are explicit and values are not shuffled between registers.
The MSO part guesses a \emph{run skeleton}: second-order variables label each
position by the chosen transition rule and enforce control consistency.
The remaining challenge is to express that all register constraints hold,
without violating well-formedness.
The skeleton induces, for each register-value
variable (and for the current datum), a partition of the word into
\emph{live intervals} on which that value is constant.
Weak guessing guarantees visibility: every live interval carrying a value has
some witness position inside it where that value occurs in the input data
track.
We then use $\segment{X}$ as a mechanism for \emph{choosing witnesses
interval-by-interval} rather than as a function of a universally quantified
position.
Roughly, we cut the scope into a family of disjoint ``operating intervals'',
and inside each such interval we existentially pick witness positions for the
currently-unwitnessed values and recurse.
A simple disjointness/alternation property of these operating intervals is what
makes the segmentation work: it lets us select witnesses independently on many
intervals without interference.
Once witnesses for all relevant values are fixed, the final check is easy:
we universally quantify the current position and evaluate each transition guard
by referring to the fixed witness positions, so each data atom contains only
one nested variable.
This yields an $\SMSOnoalpha$ sentence equivalent to the automaton.

    \section{Extensions and special cases}
    \label{sec:extensions}
    \subsection{Infinite words}

Fix a finite alphabet $\Sigma$ and atoms $\A$.
An \emph{infinite} word over $\Sigma\times\A$ is a sequence
$(a_i)_{i\in\N_+}$ with $a_i\in\Sigma\times\A$.
We write $(\Sigma\times\A)^\omega$ for the set of \emph{infinite data words}.
The notion of a data word structure extends verbatim to this setting by taking
the set of positions to be $\N_+$.

Register automata require no syntactic changes for $\omega$-words; only the
acceptance condition is replaced by the standard \emph{B\"uchi} condition.
For $\cA\in\NRA[\Sigma,\A]$, we define $\lang^\omega(\cA)$ as the set of
$\omega$-words admitting an infinite run that starts in an initial
configuration and visits accepting states from $Q_{\fin}$ infinitely often.

\namedparagraph{Scoped MSO for $\omega$-words.}
The semantics of $\SMSO[\Sigma,\A]$ is unchanged over infinite data word
structures: formulas are interpreted over positions $\N_+$, with scopes now
being intervals of $\N_+$ as before.
The correspondence between runs and MSO-definable \emph{run skeletons} carries
over: infinite runs of any $\cA\in\NRA[\Sigma,\A]$ can be described in
$\SMSO[\Sigma,\A]$ by enforcing the transition labelling at every position and
adding the usual B\"uchi recurrence condition.

For the converse direction, we reuse the translation from
\cref{lem:formula-translation-desired-property}: for every sentence
$\varphi_0\in\SMSO[\Sigma,\A]$ we construct a formula $F(\varphi_0)$ such that
$\lang^\omega(\varphi_0)$ is the existential track projection of
$\lang^\omega(F(\varphi_0)) \cap
\lang^\omega\!\bigl(\bigwedge_{x\in\VarIni}\eta_x\bigr)$, and moreover
$F(\varphi_0)$ belongs to a finite-alphabet MSO fragment (over the expanded
alphabet $\Sigma''$).
Hence, by the $\omega$-regular analogue of the
B\"uchi--Elgot--Trakhtenbrot theorem~\cite{RICHARDBUCHI19661}, $F(\varphi_0)$ is
recognised by a B\"uchi automaton on $\Sigma''$.
Following the same projection-of-product construction as in
\cref{subsec:from-formulas-to-automata}, we obtain an
$\NRA[\Sigma,\A]$ with B\"uchi acceptance that recognises
$\lang^\omega(\varphi_0)$.
In particular:

\begin{claim}
    For every finite $\Sigma$ and atoms $\A$,
    \[
        \lang^\omega{\SMSO[\Sigma,\A]}
        \;=\;
        \lang^\omega{\NRA[\Sigma,\A]}.
    \]
\end{claim}

\subsection{Atoms requiring strong guessing}

Weak guessing is essential in our main logical characterisation: $\SMSO$ can
only refer to data values that occur in the input word, whereas an NRA with
\emph{strong} guessing may introduce fresh data values that never appear on the
input track.
To accommodate such behaviour, one can enrich the multi-track variant of
$\SMSO$ (used in \cref{lem:logic-reduces-to-NRA}) with top-level existential
quantification over fresh \emph{data tracks}.
Formally, extend the syntax with quantifiers of the form $\exists\,t[\cdot]\,
\varphi$, where $t[\cdot]$ is a new track-accessor term ranging over $\A$.

\smallskip\noindent
\emph{Semantics.}
Let $\vu\colon A\to\Sigma^L$ and $\vv\colon B\to\A^L$ be families of tracks of
equal length $L$.
Then
\[
    \modelConv{\vu,\vv},\nu,I \models \exists\,t[\cdot]\ \varphi
    \quad\text{iff}\quad
    \exists\,v'\in\A^L\ \text{such that}\
    \modelConv{\vu,\Conv{\vv,v'}},\nu,I \models \varphi,
\]
where inside $\varphi$ the term $t[x]$ is interpreted as $v'[\nu(x)]$ for each
first-order variable~$x$.
Intuitively, $\exists\,t[\cdot]$ equips the logic with a supply of fresh data
values, mirroring strong guesses performed by the automaton.

\subsection{NRA without guessing}

A simple refinement of the logic suffices to capture the formalism of \NRAnoguess.
We write $\SMSOnoguess[\Sigma,\A]$ for the variant obtained by replacing the
scope modality with a \emph{scoping quantifier} $\scope{X,y}$ and adapting the
notion of top-level variables accordingly:
\begin{itemize}
    \item $\bW,\nu,I \models \scope{X,y}\,\varphi$ iff for every subscope
    $I'=\range{a}{b}\in\Subscopes{I,\nu(X)}$ we have
    $\bW,\nu[y\gets a],I' \models \varphi$;
    \item a first-order variable is \emph{top-level} if it is bound by some
    $\scope{X,y}$ that does not occur under negation in the syntax tree, and it
    is \emph{nested} otherwise;
    \item the well-formedness condition on data atoms is unchanged: every
    occurrence of $R(\v x)$ may contain at most one nested variable.
\end{itemize}
As before, scopes are intended to correspond to live intervals in the run.
The variable $y$ bound by $\segment{X,y}$ pins each register value to the left endpoint of the live interval.
Intuitively, this matches the ``no-guessing'' discipline: values may be copied
forward from earlier positions (or taken from $\cur$), but cannot be guessed
out of thin air.

    \bibliography{bibliography}

    \newpage
    \appendix

    \section{Expressions--automata equivalence}
    \label{sec:expressions--automata-equivalence}
    For the rest of this section, fix a finite alphabet $\Sigma$ and atoms $\A = (U, R_1, \ldots, R_\ell)$.

        \subsection{From expressions to automata}
        \label{subsec:from-expressions-to-automata}
        In this section we prove one direction of \cref{thm:expressions-NRA-equivalence}.

\begin{lemma}
    \label{lem:expressions-reduce-to-NRA}
    For every data regular expression $E \in \DRE[\Sigma,\A]$, there exists an
    $\cA \in \NRA[\Sigma,\A]$ such that $\lang{E} = \lang{\cA}$.
\end{lemma}

\namedparagraph{NRAs with $\bm{\varepsilon}$-transitions.}
The base model of \NRAs~(as in \cref{subsec:a-formalism-for-data-word-languages}) reads exactly one input symbol per step.
For the constructions below, it is convenient to allow $\epsilon$-transitions
that do not consume input.
    A nondeterministic register automaton with $\varepsilon$-transitions ($\NRA_\varepsilon[\Sigma, \A]$) is a tuple $\cA = (Q,\Sigma,\A,\cR,Q_\ini,Q_\fin,\delta,\delta_\varepsilon)$ where $(Q,\Sigma,\A,\cR,Q_\ini,Q_{\fin},\delta)$ is an \NRA as before and
    $\delta_\varepsilon \subseteq Q\times \constr_\A{\cR} \times Q$
    is a finite set of \emph{$\varepsilon$-transition rules} with constraints not using variable $\cur$.
    We write $(p,\sigma,\varPhi,q)\in\delta$ as $p \xRightarrow{\sigma,\varPhi} q$ and
    $(p,\varPhi,q)\in\delta_\varepsilon$ as $p \xRightarrow{\varPhi} q$.
    The induced infinite transition system uses the same configurations
    $\conf{\cA}=Q \times (\cR\to\A_{\emptyReg})$ as in \cref{subsec:a-formalism-for-data-word-languages}.
    For a data word $u=(\sigma_1,\alpha_1) \cdots (\sigma_n,\alpha_n) \in (\Sigma \times \A)^*$,
    a \emph{run} is a sequence alternating input-consuming transitions and $\varepsilon$-transitions:
    \begin{itemize}
        \item \emph{input-consuming transitions:} $(p,\mu)\xrightarrow{\sigma,\alpha}(q,\mu')$ exists iff
        $p \xRightarrow{\sigma,\varPhi} q \in \delta$ and $\nu(\mu,\alpha,\mu') \models \varPhi$,
        \item \emph{$\varepsilon$-transitions:} $(p,\mu)\xrightarrow{\varepsilon}(q,\mu')$ exists iff
        $p \xRightarrow{\varPhi} q \in \delta_\varepsilon$ and $\nu(\mu, \mu')\models\varPhi$, where $\nu(\mu,\mu')$ is defined similarly to $\nu(\mu,\argumentDot,\mu')$ but does not specify the value for $\cur$.
    \end{itemize}
    Acceptance is defined as usual: $u$ is accepted if there is a run from $(q_\ini,\mu_\ini)$ for some $q_\ini \in Q_\ini$
    labelled by $u$ that ends in some $(q_\fin, \mu)$ for $q_\fin \in Q_\fin$.

\begin{lemma}[$\varepsilon$-elimination]
    \label{lem:eps-elim}
    Fix $N \in \N$.
    For every $\cA \in \NRA_\varepsilon[\Sigma,\A]$ that executes at most~$N$ $\varepsilon$-transitions in a row there exists $\cB\in\NRA[\Sigma,\A]$ with $\lang{\cA}=\lang{\cB}$.
\end{lemma}

\begin{proof}[Proof sketch]
    This is the standard $\varepsilon$-elimination construction, applied at the level of transition rules.
    A bound $N \in \N$ on the number of $\varepsilon$-steps executed in a row is required, because the register valuations visited between $\varepsilon$-transitions need to be stored in additional registers.

    Intuitively, one replaces any pattern ``$N$ $\varepsilon$-steps, then one input step, then $N$ $\varepsilon$-steps''
    by a single input step whose constraint is the (finite) conjunction of the involved constraints with
    fresh copies of registers for intermediate valuations; the intermediate register values are existentially guessed
    by the automaton.
\end{proof}

Hence, it suffices to construct \NRAs with $\varepsilon$-transitions, ensuring they are used a bounded number times in a row between input-consuming steps; \cref{lem:eps-elim} then yields an \NRA.

\begin{proof}[Proof of \cref{lem:expressions-reduce-to-NRA}]
    We define a mapping $A$ that compiles a data regular expression $E$ into $A(E) \in \NRA_\varepsilon[\Sigma, \A]$ by structural induction on $E$, ensuring that $\lang{\cA_E}=\lang{E}$.
    Applying \cref{lem:eps-elim} yields an \NRA as required.

    Throughout, we use the convention that every transition constraint $\varPhi\in\constr_\A{\cR}$
    implicitly ensures that values of registers $r\in\cR$ not mentioned in $\varPhi$ are preserved, i.e., $\valpost{r} = \valpre{r}$.

    \proofCase{1}{$E=\emptyset$}
    Let $A(E)$ have no accepting states.
    Then $\lang{A(E)}=\emptyset=\lang{E}$.

    \proofCase{2}{$E=[\Conv{w, \varpi}]$}
    Let $\cR=\set{r_1,\ldots,r_n}$.
    Define $A(E)=(Q,\Sigma,\A,\cR,\set{0},Q_{\fin},\delta,\emptyset)$ with $Q=\range{0}{n}$ and $Q_{\fin}=\set{n}$.
    For each $i\in\rangeOne{n}$, put the rule $i-1 \xRightarrow{\sigma_i,\varPhi_i} i$ where
    \begin{align*}
        \varPhi_i \coloneqq \bigl(\valpost{r_i} = \cur\bigr)
        \land
        \begin{cases}
            \top & \text{if $i < n$,}\\
            \textstyle\bigwedge_{\psi \in \type_\A{\varpi}} \psi(\valpost{r_1},\dots,\valpost{r_n}) & \text{if $i = n$.}
        \end{cases}
    \end{align*}
    Clearly, $A(E)$ reads exactly $n$ input symbols with $\Sigma$-track $\sigma_1\cdots\sigma_n$,
    stores the data track into $(r_1,\ldots,r_n)$, and finally checks that the stored tuple lies in
    $\qfregion_\A{\varpi}$.
    Therefore $\lang{A(E)}=\qfregion_\A{\Conv{w, \varpi}}=\lang{E}$.

    \proofCase{3}{$E=E'+E''$} This case is immediate using the standard product construction.

    \proofCase{4}{$E=E'\cdot_k E''$}
    By \cref{lem:eps-elim}, we assume that automata for subexpressions do not use $\varepsilon$-transitions.
    Let $\cA'=A(E')=(Q',\Sigma,\A,\cR',Q'_\ini,Q'_{\fin},\delta')$ and
    $\cA''=A(E'')=(Q'',\Sigma,\A,\cR'',Q''_\ini,Q''_{\fin},\delta'')$,
    with $\cR',\cR''$ disjoint.
    Let $\cR''' = \set{s_1,\dots,s_k}$ be fresh registers for intended to store the data part of the guessed overlap word.
    We build $A(E) = (Q, \Sigma, \A, \cR, Q_\ini, Q_\fin, \delta, \delta_\varepsilon)$ that operates in four phases:
    \begin{enumerate}[label={(\texttt{{P\arabic*}})}, ref={(\texttt{{P\arabic*}})}, leftmargin=*,topsep=\smallskipamount, partopsep=0pt]
        \item simulates $\cA'$ on the real input prefix $u$,
        \item guesses an overlap word $v\in(\Sigma\times \A)^k$ using $\varepsilon$-steps and continues the simulation of $\cA'$ on $v$ until it is accepting, and then
        \item simulates $\cA''$ starting by feeding it $v^R$ via $\varepsilon$-steps, and finally
        \item runs $\cA''$ on the remaining real input suffix $w$.
    \end{enumerate}

    Let $\Sigma^{\le k} = \bigcup_{j=0}^k \Sigma^j$.
    Define the state space
    \begin{align*}
        Q \coloneqq \
        (Q'\times\set{\mathtt{P1}})
        &\disjointUnion
        (Q'\times\range{0}{k}\times \Sigma^{\le k}\times\set{\mathtt{P2}}) \disjointUnion {}\\
        (Q''\times\set{\mathtt{P4}})
        &\disjointUnion
        (Q''\times\range{0}{k}\times \Sigma^k\times\set{\mathtt{P3}})\,.
    \end{align*}
    Intuitively:
    \begin{itemize}
        \item $(q,\mathtt{P1})$ and $(q,\mathtt{P4})$ simulate $\cA'$ and $\cA''$, respectively, on real input,
        \item $(q,j,v,\mathtt{P2})$ has guessed $j$ overlap letters, remembers $v \in \Sigma^j$,
        and simulates $\cA'$ on the guessed overlap, and
        \item $(q,j,v,\mathtt{P3})$ replays the remembered $v\in\Sigma^k$ in reverse to simulate $\cA''$.
    \end{itemize}
    Set $Q_\ini = \set{(q'_\ini,\mathtt{P1}) \suchthat q_\ini \in Q'_\ini}$ and $Q_{\fin}=\set{(q_\fin,\mathtt{P4})\suchthat q_\fin\in Q''_{\fin}}$.
    The register set is $\cR=\cR' \disjointUnion \cR'' \disjointUnion \cR'''$.
    We construct $\delta$ and $\delta'$ as follows.
    \begin{itemize}
        \item Phases $\mathtt{P1}$ and $\mathtt{P4}$ are straightforward to implement:
        for rule \smash{$p\xRightarrow{\sigma,\varPhi} q\in\delta'$} add the input rule $(p,\mathtt{P1}) \xRightarrow{\sigma,\varPhi} (q,\mathtt{P1})$, where $\varPhi$ is read as a constraint over $\cR$ (it only mentions registers in $\cR'$).
        We add analogous transition rules for $\cA''$ and $\mathtt{P4}$.
        \item Switch from $\mathtt{P1}$ to $\mathtt{P2}$.
        For every $q\in Q'$ add the $\varepsilon$-rule $(q,\mathtt{P1}) \xRightarrow{\top} (q,0,\varepsilon,\mathtt{P2})$.
        \item Phase $\mathtt{P2}$.
        For each $j\in\range{0}{k-1}$, each $v\in\Sigma^j$, each $\sigma\in\Sigma$, and each rule
        $p\xRightarrow{\sigma,\varPhi} q\in\delta'$, let $\varPhi \equiv \varPhi[\cur \gets \valpost{s}_{j+1}] \land (\valpost{s}_{j+1} = \valpre{s}_{j+1} \lor \valpost{s}_{j+1} \neq \valpre{s}_{j+1})$ and add the $\varepsilon$-rule
        \begin{align*}
        (p,j,v,\mathtt{P2}) \xRightarrow{\varPhi'} (q,j+1,v\sigma,\mathtt{P2})\,.
        \end{align*}
        In particular, this allows the automaton to guess a fresh value of $s_{j+1}$ and remembers the guessed $\Sigma$-letter in the state.
        \item Switch from $\mathtt{P2}$ to $\mathtt{P3}$:
        After $k$ guesses, require acceptance of $\cA'$ and jump to replay phase:
        for every $v\in\Sigma^k$, every $q'_\fin\in Q'_{\fin}$ and $q''_\ini \in Q''_\ini$, and every $\sigma\in\Sigma$, add $(q'_\fin,k,v,\mathtt{P2}) \xRightarrow{\top} (q''_\ini,0,v,\mathtt{P3})$.
        \item Phase $\mathtt{P3}$: replay $v^R$ to start $\cA''$.
        Fix $v=\sigma_1\cdots\sigma_k\in\Sigma^k$.
        For each $j\in\range{0}{k-1}$ and each rule $p\xRightarrow{\sigma_{k-j},\varPhi} q\in\delta''$,
        let $\varPhi' \equiv \varPhi[\cur \gets \valpre{s}_{k-j}]$ add the $\varepsilon$-rule
        \begin{align*}
        (p,j,v,\mathtt{P3})
            \mathbin{\smash{\xRightarrow{\varPhi}}}
            (q,j+1,v,\mathtt{P3}).
        \end{align*}
        (Thus the $j$-th replay step forces the guessed datum $\cur$ to equal the stored overlap datum.)
        \item Switch from $\mathtt{P3}$ to $\mathtt{P4}$:
        for every $q\in Q''$ and every $\sigma\in\Sigma$, add $(q,k,v,\mathtt{P3}) \xRightarrow{\top} (q,\mathtt{P4})$.
    \end{itemize}

    \proofDir{Correctness.}
    We show that $\lang{A(E)}=\lang{E' \cdot_k E''}$.

    \proofDir{Direction ``$\bm{\subseteq}$''.}
    Let $uw \in \lang{A(E)}$ and fix an accepting run such that
    is $u$ the maximal prefix of $uw$ consumed while the run stays in the $\mathtt{P1}$-component.
    The run then takes the switch and performs exactly $k$ $\varepsilon$-transitions in $\mathtt{P2}$,
    thereby fixing some $\sigma_1\cdots\sigma_k\in\Sigma^k$ and storing some data values
    $\alpha_1,\dots,\alpha_k\in \A$ in $s_1,\dots,s_k$.
    Let $v=(\sigma_1,\alpha_1)\cdots(\sigma_k,\alpha_k)\in(\Sigma\times \A)^k$.
    By construction, the projection of the run to the $\cR'$-registers witnesses that $\cA'$ has an accepting run on $uv$,
    hence $uv\in\lang{\cA'}=\lang{E'}$ by the induction hypothesis.
    Next, during $\mathtt{P3}$ the automaton feeds $v^R$ (via $\varepsilon$-transitions) into $\cA''$ and then consumes $w$
    while simulating $\cA''$ on real input, ending in a final state.
    Thus $v^R w\in\lang{\cA''}=\lang{E''}$.
    Therefore, by definition of $k$-contracting concatenation, $x=uw\in\lang{E'}\cdot_k\lang{E''}=\lang{E'\cdot_k E''}$.

    \proofDir{Direction ``$\bm{\supseteq}$''.}
    Conversely, let $uw\in\lang{E'\cdot_k E''}$ such that $uv\in\lang{E'}$ and $v^R w\in\lang{E''}$ for some $v\in(\Sigma\times \A)^k$.
    By the induction hypothesis, $\cA'$ accepts $uv$ and $\cA''$ accepts $v^R w$.
    We build a run of $A(E)$ over $uw$ in a straightforward way.

    \proofCase{5}{$E=F^{k,*}$}
    Let $K=\lang{F}$.
    Recall $\lang{F^{k,*}}=\bigcup_{n\ge 0} K^{k,n}$ where $K^{k,0}$ is the set of palindromes of length $2k$.

    \proofDir{Automaton for $K^{k,0}$.} Observe that the language of palindromes of length $2k$ are definable by an expression $[w_1] + [w_2] + \cdots + [w_m]$ by enumerating representatives $w_1, \ldots, w_m$ of all orbits of words of length $2k$, thus by the inductive hypothesis it has a corresponding automaton $\cA_0$.

    \proofDir{Automaton for $\bigcup_{n\ge 1} K^{k,n}$.}
    This is an iterated variant of the construction for $E \cdot_k F$.
    We construct $\cA_{\ge 1} \in \NRA_\varepsilon[\Sigma, \A]$ that starts simulating $A(F)$ on real input, and can nondeterministically decide to  perform the same ``guess+replay'' gadget as in the case
    $E \cdot_k F$, but now from $A(E)$ it returns to the simulation in $A(E)$.
    Concretely, take the above construction for $A{F \cdot_k F}$ and identify states $(q, \mathtt{P1})$ and $(q, \mathtt{P4})$.
    This yields an automaton that recognizes $\bigcup_{n\ge 1} K^{k,n}$.
    The correctness argument is the same as above, by counting the number of completed gadgets in the run:
    each completed ``guess+replay'' gadget corresponds to one application of $\cdot_k$, and the final accepting segment corresponds to the last factor.

    Finally, define $A(E)$ as the union (case $+$) of $\cA_0$ and $\cA_{\ge 1}$.
    Then $\lang{\cA_E}=\lang{F^{k,*}}$.

    This completes the inductive construction, hence $\lang{A(E)}=\lang{E}$ for every $E$.
    Applying \cref{lem:eps-elim} yields an \NRA without $\varepsilon$-rules recognizing the same language.
\end{proof}

        \subsection{From automata to expressions}
        \label{subsec:from-automata-to-expressions}
            In this section we prove the converse direction of \cref{thm:expressions-NRA-equivalence}.

    \begin{lemma}
        \label{lem:nra-reduce-to-expressions}
        Let $\cA \in \NRA[\Sigma,\A]$.
        There exists a data regular expression $E(\cA)\in\DRE[\Sigma,\A]$ such that $\lang{\cA}=\lang{E(\cA)}$.
    \end{lemma}

    \begin{proof}
        Fix an $\cA=(Q,\Sigma,\A,\cR,Q_\ini,Q_{\fin},\delta) \in \NRA_k[\Sigma, \A]$.
        Fix once and for all an enumeration $\cR=\set{r_1,\dots,r_k}$.
        Also fix a distinguished \emph{padding letter} $\pad\in\Sigma$ (assume $\Sigma\neq\emptyset$).

        \namedparagraph{The transition-label NFA.}
        Let $\Gamma\coloneqq\delta$ be the finite set of transition rules of $\cA$.
        Define $\cB = (Q,\Gamma,Q_\ini,Q_{\fin},\delta') \in \NFA[\Gamma]$,
        where $\delta' \subseteq Q \times\Gamma\times Q$ contains the transition
        $p \xrightarrow{t} q$ iff $t\in\delta$ is the rule $p \xRightarrow{\sigma,\varPhi} q$ for some $\sigma\in\Sigma$ and $\varPhi\in\constr_\A{\cR}$.
        Thus $\cB$ reads words over alphabet $\Gamma$ and accepts exactly the \emph{control-consistent} sequences of rules:
        \[
            \lang{\cB}=\{\pi \in\delta^* \mid \text{$\pi$ is a path from some $q_\ini \in Q_\ini$ to some $q_\fin \in Q_{\fin}$}\}.
        \]
        By Kleene's theorem (\cref{thm:kleene-regular}), there exists a (classical) regular expression $E_\cB\in\RE[\Gamma]$ such that $\lang{E_\cB}=\lang{\cB}$.

        \namedparagraph{Unions of regions.}
        Observe that the set of quantifier-free regions of words in $\A^m$ is finite.
        Let $\Rep_m$ be the set of representatives of these regions.
        Let $\psi(x_1,\dots,x_m)$ be a quantifier-free first-order formula over $\A$,
        with free variables among $x_1,\dots,x_m$.
        For $w \in \Sigma^m$ define $\any{m,\psi}\in\DRE[\Sigma,\A]$ as the union of quantifier-free regions~satisfying~$\psi$
        \begin{align*}
            \any{w,\psi}
            \coloneqq
            \textstyle\sum_{\substack{\varpi \in \Rep_m \\ \A\models\psi(\varpi)}}
            \Bigl[\Conv{w, \varpi}\Bigr].
        \end{align*}
        Finally, define \emph{guard} $G \coloneqq \any{\pad^k,\top}$.

        \namedparagraph{Encoding transition rules.}
        Let $t\in\delta$ be a rule $p \xRightarrow{\sigma,\varPhi} q$.
        We convert $\varPhi$ into a formula $\hat{\varPhi}(x_0,\dots,x_{2k})$ by substituting: $\valpre{r_i}\mapsto x_{k-i},\ \cur\mapsto x_{k},\ \valpost{r_i}\mapsto x_{k+i}$ for all $i\in\rangeOne{k}$, and replacing every atomic formula $\emptyReg(r)$ by $\top$.
        Intuitively, the first $k$ variables encode, in reverse order, the register valuation before taking $t$, and the last $k$ encode the valuation after taking $t$.
        Initial live intervals extend to the very beginning of the encoding, eliminating the need for $\emptyReg$.
        Define the data expression associated with $t$ as
        \[
            \Block(t) \coloneqq \any{\pad^k\sigma\pad^k,\hat{\varPhi}}\ \in\DRE[\Sigma,\A].
        \]
        Intuitively, $\Block(t)$ consists of all $(2k+1)$-letter blocks whose middle symbol has finite letter $\sigma$,
        whose padding letters are $\pad$, and whose data values satisfy the constraint $\varPhi$ when interpreted
        as encoding of some valuation $\mu$, current data value, and encoding of some $\mu'$.

        \namedparagraph{Interpreting $\bm{\varepsilon}$ as the neutral element of $\bm{(\cdot_k)}$.}
        We will translate $E_\cB$ into a $k$-contracting data expression; for this we need an expression that plays the role of $\varepsilon$.
        Define the formula
        \[
            \pal_k{x_1,\dots,x_{2k}} \coloneqq \textstyle\bigwedge_{i\in\rangeOne{k}} x_i = x_{2k-i+1}
        \]
        and set $\Id(k) \coloneqq \any{\pad^{2k},\pal_k{}}$.
        It is easy to see that $\lang{\Id(k)}$ is the set
        \begin{align*}
            \set{\Conv{\pad^k,\varpi}\Conv{\pad^k,\varpi}^R \suchthat \varpi\in\A^k}\,.
        \end{align*}

        \namedparagraph{Translating $\bm{E_\cB}$ into a data regular expression.}
        Define a translation $\Tr(\cdot)$ from classical regular expressions over $\Gamma$ to data expressions in $\DRE[\Sigma,\A]$ by:
        \begin{align*}
            \Tr(\emptyset) &\coloneqq \emptyset & \Tr(E+F) &\coloneqq \Tr(E)+\Tr(F)\\
            \Tr(\epsilon) &\coloneqq \Id(k) & \Tr(E\cdot F) &\coloneqq \Tr(E)\cdot_k \Tr(F)\\
            \Tr(t) &\coloneqq \Block(t)\quad(t\in\Gamma) & \Tr(E^*) &\coloneqq \Tr(E)^{k,*}\,.
        \end{align*}
        Finally, let $E(\cA) \coloneqq G \cdot_k \Tr(E_\cB) \cdot_k G$.

        \namedparagraph{Correctness.}
        We prove both inclusions of $\lang{E(\cA)}=\lang{\cA}$.

        \proofDir{Direction $\lang{\cA}\subseteq\lang{E(\cA)}$.}
        Let $w=(\sigma_1,\alpha_1)\cdots(\sigma_n,\alpha_n)\in(\Sigma\times \A)^*$ be accepted by $\cA$.
        Fix an accepting run
        \[
            (Q_\ini,\mu_0)\xrightarrow{\sigma_1,\alpha_1}(q_1,\mu_1)\xrightarrow{\sigma_2,\alpha_2}\cdots
            \xrightarrow{\sigma_n,\alpha_n}(q_n,\mu_n),
        \]
        where $\mu_0(r)=\emptyReg$ for all $r\in\cR$ and $q_n\in Q_{\fin}$.
        Fix the corresponding sequence $t_1 t_2 \cdots t_n$ of transition rules.
        By construction, it belongs to $\lang{\cB} = \lang{E_\cB}$, and for every $i \in \rangeOne{n}$ we have
        \begin{align*}
            t_i = (q_{i-1}\xRightarrow{\sigma_i,\varPhi_i} q_i)
            \quad\text{with}\quad
            \nu(\mu_{i-1},\alpha_i,\mu_i)\models\varPhi_i.
        \end{align*}
        Consider, for each $i$, the $(2k+1)$-block
        \begin{align*}
            B_i \coloneqq
            \Conv{\pad^k\,\sigma_i\,\pad^k,\ \
                \mu_{i-1}(r_k)\cdots\mu_{i-1}(r_1)\ \alpha_i\ \mu_i(r_1)\cdots\mu_i(r_k)}\,.
        \end{align*}
        By construction and the definition of $\hat{\varPhi_i}$, we have $B_i\in\lang{\Block(t_i)}$.
        Moreover, the $k$-contracting concatenation precisely matches post- and pre-valuations:
        the last $k$ symbols of $B_i$ are
        $\Conv{(\pad^k,\mu_i(r_1)\cdots\mu_i(r_k))}$,
        which, after reversal, equals the first $k$ symbols of $B_{i+1}$.
        Hence the word $W \coloneqq B_1 \cdot_k B_2 \cdot_k \ldots \cdot_k B_n$ is well-defined and satisfies $W\in\lang{\Block(t_1)\cdot_k\cdots\cdot_k\Block(t_n)}\subseteq\Tr(E_\cB)$.
        A direct computation of $k$-contractions yields
        \[
            W =
            \Conv{\pad^k,\varpi'} \cdot
            w \cdot
            \Conv{\pad^k,\varpi'}
        \]
        for some $\varpi', \varpi'' \in \A^k$.
        By definition of the guard $G$, $w\in \lang[\big]{G \cdot_k \Tr(E_\cB) \cdot_k G}=\lang{E(\cA)}$.

        \proofDir{Direction $\lang{E(\cA)}\subseteq\lang{\cA}$.}
        Let $w\in\lang{E(\cA)}$.
        By the semantics of $\cdot_k$, there exist words
        $u, v\in\lang{G}$ and $W\in\Tr(E_\cB)$
        such that $w = u\cdot_k W \cdot_k v$.
        Since $\lang{G} = (\set{\pad}\times \A)^k$,
        the effect of the two outer $k$-contractions is to remove a length-$k$ prefix and a length-$k$ suffix of $W$.
        In particular, $W = u^R w v^R$.

        Now, since $W\in\Tr(E_\cB)$, unfolding the definition of $\Tr(\cdot)$ yields
        \begin{align*}
            W = u_0^R a_1 u_1\ \  \cdot_k\ \  u_1^R a_2 u_2 \ \ \cdot_k\ \  u_2^R \cdots u_{n-1}\ \  \cdot_k\ \  u_{n-1}^R a_n u_n
        \end{align*}
        for some $u_i \in (\set{\pad} \times \A)^k$ such that $B_i \coloneqq u_{i-1}^R a_i u_i \in \lang{\Block(t_i)}$ for $i \in \rangeOne{n}$.
        and some sequence $t_1 \cdots t_n \in \lang{\cB}$ of transition rules beginning in some initial $q_{\ini} \in Q_\ini$ and ending in some accepting $q_\fin \in Q_\fin$.
        Each block $B_i$ determines (by its data track) a triple of register valuations and a current datum
        $(\mu_{i-1},\alpha_i,\mu_i)$ such that $\nu(\mu_{i-1},\alpha_i,\mu_i)\models\varPhi_i$ for the constraint of $t_i$ (potentially after replacing some data values by $\emptyRegElement$).

        Thus we obtain a well-defined accepting run of $\cA$ on the word $w = a_1 a_2 \cdots a_n$:
        \[
            (q_0,\mu_0)\xrightarrow{\sigma_1,\alpha_1}(q_1,\mu_1)\cdots\xrightarrow{\sigma_n,\alpha_n}(q_n,\mu_n)\,,
        \]
        where $\mu_0$ is the initial valuation, $q_0 = q_{\ini}$ and $q_n = q_\fin$.
        Hence $w\in\lang{\cA}$.
    \end{proof}

    \section{Logic--automata equivalence}
    \label{sec:logic--automata-equivalence}
    For the rest of this section, fix a finite alphabet $\Sigma$ and atoms $\A = (U, R_1, \ldots, R_\ell)$.

        \subsection{From formulas to automata}
        \label{subsec:from-formulas-to-automata}
        In this section we prove one direction of \cref{thm:logic-NRA-equivalence},
captured by the following lemma.

\begin{lemma}
    \label{lem:logic-reduces-to-NRA}
    For every sentence $\varphi_0 \in \SMSO$, there exists an
    $\cA \in \NRA[\Sigma,\A]$ such that $\lang{\varphi_0} = \lang{\cA}$.
\end{lemma}

For the rest of the section fix a sentence $\varphi_0 \in \SMSO$.
Without loss of generality, all variables bound by quantifiers in $\varphi_0$
are pairwise distinct.
Let $\Var$ be their set and write
$\Var = \VarIni \disjointUnion \mkern2mu\VarOther$,
where $\VarIni$ consists of the top-level variables.
Set $N \coloneqq |\VarIni|$.

\begin{remark}
    \label{rem:at-most-one-nested}
    By \cref{cond:SMSO-condition}, every atomic formula of the form
    $R(x_1,\ldots,x_r)$ in $\varphi_0$ contains at most one variable from
    $\VarOther$.
\end{remark}

\namedparagraph{A multi-track alphabet.}
In this section, we work with formulas of $\SMSO[\Sigma',\A']$, where
\begin{align*}
    \Sigma' &\coloneqq \Sigma \times \{\ubegin,\utail\}^N,
    &
    \A' &\coloneqq \A^{(N+1)}.
\end{align*}
Here, $\A^{(N+1)}$ denotes the structure with universe $U^{N+1}$ and an expanded
signature: for every relation symbol $R$ of $\A$ of arity $r$ and every tuple
$(i_1,\ldots,i_r) \in \rangeOne{N+1}^r$, the signature of $\A'$ contains a
relation symbol $R^{(i_1,\ldots,i_r)}$ that tests $R$ on the corresponding
coordinates.

Given a word $\Conv{w, \varpi} \in (\Sigma \times \A)^*$ and families of words
$\vu\colon \VarIni \to \B^*$ and $\vv\colon \VarIni \to \A^*$ of the same length,
we write $\modelConv{w,\vu,\varpi,\vv}$ for the data word structure
$\modelConv{\Conv{w,\vu},\Conv{\varpi,\vv}}$.
It is convenient to view such structures as \emph{multi-track} words.

We refer to the components of $\Sigma'$ and $\A'$ as follows:
\begin{itemize}
    \item finite tracks: $w$ and $u_x$ for $x\in\VarIni$,
    \item infinite tracks: $\varpi$ and $v_x$ for $x\in\VarIni$.
\end{itemize}
Assuming fixed bijections between track names and tuple indices, we use the
following syntactic sugar:
\begin{align*}
    R(t_1[x_1],\ldots,t_r[x_r])
    &\equiv R^{(i_1,\ldots,i_r)}(x_1,\ldots,x_r), &
    t[x] = a
    &\equiv a_1(x)\lor\cdots\lor a_n(x),
\end{align*}
where each $t_j$ is the track corresponding to index $i_j$, $t$ is a finite
track, $a\in\Sigma\cup\{\ubegin,\utail\}$, and $\{a_1,\ldots,a_n\}$ is the set of
letters of $\Sigma'$ whose coordinate for $t$ equals~$a$.

This notation can be read either as syntactic sugar for $\SMSO[\Sigma',\A']$ or,
equivalently, as a mild multi-sorted extension with \emph{track-accessor} terms
$t[x]$.
We freely switch between these two viewpoints.

\namedparagraph{Preparatory definitions.}
For every $x \in \VarIni$ define the formulas
\begin{align*}
    \eta_x &\coloneqq
    {\underbrace{\forall{s,t}
    (\functionT{Succ}{s,t} \land u_x[t]=\utail)
    \rightarrow
    v_x[s]=v_x[t]}_{\substack{\text{For all $s,s+1$, if $u_x[s+1]=\utail$ then $v_x$ is constant.}}}} \land
    {\underbrace{\forall{s} \functionT{First}{s} \rightarrow u_x[s] = \ubegin}_{\text{The track $u_x$ starts with $\ubegin$.}}}\,,
    \\[2pt]
    \theta_{x,S} &\coloneqq {\underbrace{v_x[x] = \varpi[x]}_{
        \mathclap{\substack{\text{The value $\varpi[x]$ is also} \\ \text{present on the $v_x$ track.}}}
    }} \land {\underbrace{\forall{s} S(s) \rightarrow (u_x[s] = \ubegin \leftrightarrow \functionT{FirstIn}{s, S})}_{
        \substack{\text{On the scope $S$, the track $u_x$ forms}\\\text{a word in $\ubegin \utail^*$.}}
    }}\,.
\end{align*}
where $\functionT{Succ}{s,t} \equiv \neg \exists{p}. s < p \land p < t$, $\functionT{First}{s} \equiv \neg \exists{p}. p < s$ and $\functionT{FirstIn}{s, S} \equiv \neg \exists{p}. S(p) \land p < s$.
We call a formula $\varphi$ \emph{locally testing} if every data atomic formula
present in $\varphi$ is of the form $R(t_1[x],\ldots,t_r[x])$ for some variable
$x$; that is, $\varphi$ never compares data values from different positions of the
word.

\begin{lemma}[Subscopes are MSO-definable]
    \label{lem:subscope-is-mso-definable}
    There is an \MSO formula $\functionT{Subs}{S',X,S}$ such that for every data word structure $fS$
    and valuation $\nu$ the satisfaction relation
    $fS,\nu \models \functionT{Subs}{S',X,S}$ holds if and only if
    $\nu(S)$ is a nonempty range and $\nu(S') \in \Subscopes{\nu(S), \nu(X)}$.
\end{lemma}

\namedparagraph{Formula translation.}
    For the proof of \cref{lem:formula-translation-desired-property}, we define a
    translation $F(\varphi)$ that maps each subformula $\varphi$ of $\varphi_0$
    to a formula of $\SMSO[\Sigma',\A']$ that
    \begin{itemize}
        \item is locally testing, and
        \item does not contain the scope modality $\segmentSymbol$.
    \end{itemize}
    We define $F$ by structural recursion on~$\varphi$, using a helper function
    $G(\varphi, S)$ that takes an additional scope variable as a parameter.
    Consider a data atomic formula $\varphi \equiv R(x_1,\ldots,x_n)$.
    By \cref{rem:at-most-one-nested},
    $\lvert \set{x_1,\ldots,x_n}\cap\VarOther\rvert \le 1$.
    We set
    \begin{align}
        G(\varphi,S) &\coloneqq
        \begin{cases}
            R(t_{1,x},\ldots,t_{n,x}) &
            \text{if $\{x_1,\ldots,x_n\}\cap\VarOther= \set{x}$,}\\
            \exists{z}. z\in S \land R(t_{1,z},\ldots,t_{n,z}) &
            \text{if $\{x_1,\ldots,x_n\}\cap\VarOther=\emptyset$,}
        \end{cases}
        \intertext{where $z$ is fresh and $t_{i,y}$ is the track-accessor term
            $\varpi[y]$ when the variable symbols $x_i$ and $y$ coincide, and
            $v_{x_i}[y]$ otherwise. For the remaining constructs, define:}
        G(a(x),S) &\coloneqq w[x]=a,\\
        G(X(y),S) &\coloneqq X(y),\\
        G(x < y,S) &\coloneqq x < y,\\
        G(\gamma\land\gamma',S) &\coloneqq G(\gamma,S)\land G(\gamma',S),\\
        G(\gamma\lor\gamma',S) &\coloneqq G(\gamma,S)\lor G(\gamma',S),\\
        G(\neg\gamma,S) &\coloneqq \neg G(\gamma,S),\\
        G(\exists{x}\gamma,S) &\coloneqq
        \begin{cases}
            \exists{x}. x\in S \land G(\gamma,S) \land \theta_{x, S} &\text{if $x \in \VarIni$}\\
            \exists{x}. x\in S \land G(\gamma,S) &\text{otherwise,}
        \end{cases}\\
        G(\exists{X}\gamma,S) &\coloneqq \exists{X}. X\subseteq S \land G(\gamma,S),\\
        G(\segment{X}\gamma,S) &\coloneqq
        \forall{S'}. \functionT{Subs}{S', X, S} \limplies G(\gamma,S'),\quad \text{where $S'$ is a fresh SO variable.}
    \end{align}
    Finally, define
    \begin{align*}
        F(\varphi_0) \coloneqq
        \begin{cases}
        {\neg\functionT{Empty}{}} \land (\exists{S}. G(\varphi_0, S) \land \forall{y} S(y)) \lor \functionT{Empty}{} & \text{if $\varepsilon \in \lang{\varphi_0}$} \\
        {\neg\functionT{Empty}{}} \land (\exists{S}. G(\varphi_0, S) \land \forall{y} S(y)) & \text{otherwise,}
        \end{cases}
    \end{align*}
    where $S$ and $y$ are fresh variable symbols, and
    $\functionT{Empty}{} \equiv \neg\exists{x}. x = x$.

The proof of \cref{lem:logic-reduces-to-NRA} breaks down into the following
lemmas.

\begin{lemma}
    \label{lem:formula-translation-desired-property}
    A word $\Conv{w, \varpi} \in (\Sigma \times \A)^*$ belongs to $\lang{\varphi_0}$
    if and only if there exist families $\vu\colon {\VarIni} \to \B^*$ and
    $\vv\colon \VarIni \to \A^*$ such that
    $\Conv{{w, \vu, \varpi, \vv}} \in \lang{F(\varphi_0)} \cap
    \lang[\big]{\textstyle\bigwedge_{x \in \VarIni}\eta_x}$.
\end{lemma}

\begin{lemma}
    \label{lem:locally-testing-have-nras}
    For every locally testing $\varphi \in \SMSO[\Sigma', \A']$ that does not use
    the scope modality $\segmentSymbol$, there exists an
    $\cA \in \NRA_0[\Sigma', \A']$ such that $\lang{\varphi} = \lang{\cA}$.
\end{lemma}

\begin{lemma}
    \label{lem:etas-cap-nra-projection-can-be-recognised-by-an-nra}
    For every $\cA \in \NRA_0[\Sigma', \A']$
    there is an \NRA that recognises the language
    \begin{align*}
        \set{ \Conv{w, \varpi} \suchthat \exists{\vu, \vv}. \Conv{w,\vu, \varpi, \vv} \in \lang{\cA} \cap \lang[\big]{\textstyle\bigwedge_{x \in \VarIni} \eta_x}}\,.
    \end{align*}
\end{lemma}
\Cref{lem:etas-cap-nra-projection-can-be-recognised-by-an-nra} should not be surprising, as $\eta_x$ is recognisable by NRA in a straightforward way, and the construction needs to implement intersection (product construction) and existential track projection.
We decided to provide a monolithic construction to avoid technical difficulties related to presence of $\A$ and $\A'$: register automata over $\A'$ would formally need to hold tuples of atom values in registers, which would be cumbersome.
Assuming \crefrange{lem:formula-translation-desired-property}{lem:etas-cap-nra-projection-can-be-recognised-by-an-nra},
\cref{lem:logic-reduces-to-NRA} follows by a short combination argument.

\begin{proof}[Proof of \cref{lem:logic-reduces-to-NRA}, assuming \crefrange{lem:formula-translation-desired-property}{lem:etas-cap-nra-projection-can-be-recognised-by-an-nra}]
    Fix an arbitrary word $\Conv{w, \varpi} \in (\Sigma \times \A)^*$.
    By \cref{lem:formula-translation-desired-property}, we have
    $\Conv{w,\varpi} \in \lang{\varphi_0}$ if and only if there exist families
    $\vu \in {\VarIni} \to \B^*$ and $\vv \in {\VarIni} \to \A^*$ such that
    $\Conv{{w, \vu, \varpi, \vv}} \in \lang{F(\varphi_0)} \cap
    \lang[\big]{\textstyle\bigwedge_{x \in \VarIni}\eta_x}$.

    By \cref{lem:locally-testing-have-nras}, we obtain $\cA \in \NRA_0[\Sigma', \A']$ recognising $\lang{F(\varphi_0)}$.
    Finally, \cref{lem:etas-cap-nra-projection-can-be-recognised-by-an-nra} yields $\cB \in \NRA[\Sigma, \A]$ such that $\Conv{w,\varpi}\in\lang{\cB}$ if and only if $\Conv{w,\varpi}\in\lang{\varphi_0}$.
    Hence $\lang{\cD} = \lang{\varphi_0}$, as required.
\end{proof}

\namedparagraph{Proof roadmap.}
We first prove the more straightforward lemmas
\crefrange{lem:locally-testing-have-nras}{lem:etas-cap-nra-projection-can-be-recognised-by-an-nra},
and then turn to the main technical step,
\cref{lem:formula-translation-desired-property}.

\begin{proof}[\of{\cref{lem:locally-testing-have-nras}}]
    Fix a locally testing sentence $\varphi \in \SMSO[\Sigma', \A']$ that does not
    use the modality $\segmentSymbol$.
    Our goal is to construct an $\cA \in \NRA_0[\Sigma',\A']$ such that
    $\lang{\varphi} = \lang{\cA}$.
    Equivalently, for every word $\Conv{w, \varpi}$ we want
    $\Conv{w, \varpi} \in \lang{\cA}$ if and only if
    $\modelConv{w, \varpi} \models \varphi$.
    Note that all subformulas of $\varphi$ are interpreted over the initial scope
    $\PosOf{w}$.

\namedparagraph{Translation $\bm{H}$.}
    Let $\cR$ be the set of all relation names in the signature of $\A'$.
    We extend $\Sigma'$ with quantifier-free types of data values: $\Sigma'' = \Sigma' \times \powerset{\cR}$.
    Let us define a translation $H(\psi)$ that maps subformulas of $\varphi$ to formulas of $\MSO[\Sigma'']$.
    Recall that due to locally-testing assumption, all data atomic formulas test a single position variable $x$.
    We define
    \begin{align*}
        H(R^{(i_1, \ldots, i_r)}(x, \ldots, x)) &\equiv \textstyle\bigvee_{\substack{\sigma = (a, \cR') \in \Sigma''\\R^{(i_1, \ldots, i_r)} \in \cR'}} \sigma(x) &
        H(a(x)) &\equiv \textstyle\bigvee_{\sigma = (a, \cR') \in \Sigma''} \sigma(x)\,.
    \end{align*}
    and define $H$ trivially for all the remaining kinds of formulas: $H(\psi \land \gamma) \equiv H(\psi) \land H(\gamma)$, $H(\neg\psi) \equiv \neg H(\psi)$, $H(\exists{x} \psi) = \exists{x} H(\psi)$, \ldots.

    Given a data word structure $\modelConv{w, \vu, \varpi, \vv}$ of length $L$ we define a
    word $\rho(\varpi, \vv) \in \powerset{\cR}^{L}$ such that for every position $i \in \rangeOne{L}$ the letter
    $\rho[i]$ is $\set[\big]{R \in \cR \suchthat \alpha_i^{\arity{R}} \in R}$, where $\alpha_i$ is the tuple $\alpha_i = (\varpi[i], \vv[i]) \in \A'$.

    \begin{claim}
        Fix a subformula $\psi$ of $\varphi$, and an interpretation $(\modelConv{w, \vu, \varpi, \vv}, \nu)$.
        Then $\modelConv{w, \vu, \varpi, \vv},\nu \models \varphi$ if and only if $\modelConv{w, \vu, \rho(\varpi, \vv)} \models G(\varphi)$.
    \end{claim}
    We prove the claim by structural induction.
    Fix a subformula $\psi$ of $\varphi$, and an interpretation $(\modelConv{w, \vu, \varpi, \vv}, \nu)$.
    \proofCase{1}{$\psi \equiv R^{(i_1, i_2, \ldots, i_r)}(x, x, \ldots, x)$}
    Let $i = \nu(x)$.
    The following statements are equivalent: $\modelConv{w, \vu, \varpi, \vv},\nu \models \varphi$ iff $(\varpi[i], \vv[i])^{\arity{R}} \in R$ iff $R \in \rho(\varpi, \vv)[i]$ iff $\modelConv{w, \vu, \rho(\varpi, \vv)}, \nu \models G(\psi)$.

    \proofCase{2}{$\psi \equiv a(x)$}
    Let $i = \nu(x)$.
    The following statements are equivalent:
    $\modelConv{w, \vu, \varpi, \vv},\nu \models \varphi$ iff
    $w[i] = a$ iff
    $\modelConv{w, \vu, \rho}, \nu \models G(\psi)$.

    \proofCase{3-10}{otherwise}
    All the remaining cases are trivial and were omitted.

    Let $\cA = (Q, \Sigma'', Q_\ini, Q_\fin, \delta) \in \DFA[\Sigma'']$ be an automaton for $G(\varphi)$ obtained through \cref{thm:buchi-elgot-trakhtenbrot}.
    \begin{claim}
        Let $\Conv{w, \vu, \varpi, \vv}$ be any word over $(\Sigma' \times \A')^*$.
        Then $\Conv{w, \vu, \varpi, \vv} \in \lang{\varphi}$ if and only if $\Conv{w, \vu, \rho(\varpi, \vv)} \in \lang{\cA}$.
    \end{claim}
    Define $\cB = (Q, \Sigma', \A', Q_\ini, F, \delta') \in \NRA_0[\Sigma', \A']$ such that for every $\cR' \subseteq \cR$, $\sigma \in \Sigma'$ and $p \in Q$ relation $\delta'$ contains a rule
    $p \xRightarrow{\sigma, \Phi(\cR')} q$ where $\delta(p, (\sigma, \cR'))$ and
    \begin{align*}
        \Phi(\cR') = \textstyle\bigwedge_{R \in \cR'} R(\cur, \ldots, \cur) \land \bigwedge_{R \in \cR \setminus \cR'} \neg R(\cur, \ldots, \cur)
    \end{align*}
    and $\delta'$ has no other rules.
    It is not hard to see that for any $(\sigma, \alpha) \in \Sigma' \times \A'$ and $p,q \in Q$ automaton
    $\cA$ has a transition $p \xrightarrow{\sigma, \rho(\alpha))} q$
    if and only if $\cB$ has a transition $p \xrightarrow{\sigma, \alpha} q$.
    Therefore, by straightforward induction, we obtain that
    $\Conv{w, \vu, \rho(\varpi, \vv)} \in \lang{\cA}$ if and only if $\Conv{w, \vu, \varpi, \vv} \in \lang{\cB}$.
    Thus, $\cB$ is the postulated $\NRA_0[\Sigma', \A']$ recognising $\lang{\varphi}$.
\end{proof}

\begin{proof}[\of{\cref{lem:etas-cap-nra-projection-can-be-recognised-by-an-nra}}]
        Let $\cA=(Q,\Sigma',\A',\emptyset,Q_\ini,Q_\fin,\delta)$.
        Since $\cA$ has $0$ registers, each rule $p\xRightarrow{\sigma',\Phi}q$ tests only the current $\A'$-letter.
        We build $\cB$ over $\Sigma\times\A$ that (i) guesses the projected-away tracks $\vu,\vv$ on the fly,
        (ii) checks $\cA$ on the guessed letter, and (iii) enforces $\bigwedge_x\eta_x$ by remembering $v_x[i-1]$ in a register.

        \namedparagraph{Registers and the $\bm{\eta}$-guard.}
        Let $\cR\coloneqq \set{r_x \suchthat x\in\VarIni}$, intended so that after reading position $i$ we have $\mu_i(r_x)=v_x[i]$.
        For a marker tuple $\gamma=(\gamma_x)_{x\in\VarIni}$, where $\gamma_i \in \B$, define
        \begin{align*}
            \Eta(\gamma)
            \coloneqq
            \textstyle\bigwedge_{\substack{x\in\VarIni\\\gamma_x=\text{\raisebox{-0.5pt}{$\utail$}}}}
            \bigl(\valpre{r}_x=\valpost{r}_x\bigr)\,.
        \end{align*}
        This is exactly the one-step form of $\eta_x$: if $u_x[i]=\utail$ then $v_x[i]=v_x[i-1]$.

        \namedparagraph{Compiling $\bm{\A'}$-tests.}
        Fix a rule $p\xRightarrow{(a,\gamma),\Phi}q$ of $\cA$.
        Because $\cA$ is $0$-register,
        $\Phi$ is equivalent to some constraint over $\A$ whose free variables are $\cur$ (the input data value)
        and the components $v_x(\cur)$ of the $(N+1)$-tuple from $\A'$. Let $\tau(\Phi)\in\constr_\A{\cR}$ denote the corresponding
        register-constraint obtained by reading $v_x(\cur)$ as $\valpost{r}_x$.
        Thus $\tau(\Phi)$ checks $\cA$'s $\A'$-condition against the tuple guessed into the \emph{post} register values.

        \namedparagraph{Construction of $\bm{\cB}$.}
        We extend the states with a one-bit ``first position'' flag to avoid enforcing $\eta_x$ at $i=1$ (since $\eta_x$ is quantified over $i>1$).
        Let $Q'\coloneqq Q\times\set{0,1}$, initial states are $Q_\ini'\coloneqq Q_\ini \times \set{0}$, and accepting states $Q_\fin' \coloneqq Q_\fin \times \set{1}$.
        Define $\cB=(Q',\Sigma,\A,\cR,Q_\ini',Q_\fin',\delta_B)$ where for every rule $p\xRightarrow{(a,\gamma),\Phi}q$ of $\cA$ we include:
        \begin{align*}
        (p,0) \xRightarrow{a,\ \tau(\Phi)} (q,1)
        \qquad\text{and}\qquad
        (p,1) \xRightarrow{a,\ \tau(\Phi)\land \Eta(\gamma)} (q,1).
        \end{align*}
        Reading only $a\in\Sigma$ implements projection of $\vu$: automaton $\cB$ nondeterministically chooses the marker tuple $\gamma$ by choosing which rule to take.
        Projection of the data value track $v_x$ is done by guessing the register value $\valpost{r}_x$ at each step.

        \namedparagraph{Correctness.}
        We need to show that $\Conv{w,\varpi}\in\lang{\cB}$ if and only if there exist $\vu,\vv$ such that
        $\Conv{w,\vu,\varpi,\vv}\in\lang{\cA}\cap\lang{\bigwedge_x\eta_x}$.

        \proofDir{Direction ``$\bm{\Rightarrow}$''.} Fix an accepting run of $\cB$ on $\Conv{w,\varpi}$:
        \begin{align*}
            ((q_0,0),\mu_0)\xrightarrow{w[1],\varpi[1]}((q_1,1),\mu_1)\xrightarrow{}\cdots
            \xrightarrow{w[n],\varpi[n]}((q_n,1),\mu_n),
        \end{align*}
        with $\mu_0(r)=\emptyReg$ for all $r\in\cR$.
        Let $\vu[i]=\gamma_i$ be the marker tuple chosen by the transition at step $i$, and fix $\vv$ such that set $v_x[i]\coloneqq \mu_i(r_x)$.
        At each step $i$, the chosen rule in $\cB$ came from some rule
        $q_{i-1}\xRightarrow{(w[i],\vu[i]),\Phi_i}q_i$ of $\cA$, and $\tau(\Phi_i)$ holds with $\valpost{r}_x=v_x[i]$,
        so $\cA$ accepts $\Conv{w,\vu,\varpi,\vv}$ along the same state sequence.
        Moreover, for $i>1$, $\Eta(\vu[i])$ enforces $v_x[i]=v_x[i-1]$ whenever $u_x[i]=\utail$, hence
        $\Conv{w,\vu,\varpi,\vv}\models \bigwedge_x\eta_x$.

        \proofDir{Direction ``$\bm{\Leftarrow}$''.} Conversely, assume $\Conv{w,\vu,\varpi,\vv}\in\lang{\cA}\cap\lang{\bigwedge_x\eta_x}$ and fix an accepting run of $\cA$
        \begin{align*}
            q_0 \xRightarrow{(w[1],\vu[1]),\Phi_1} q_1 \xRightarrow{}\cdots \xRightarrow{(w[n],\vu[n]),\Phi_n} q_n.
        \end{align*}
        Define register valuations by $\mu_0(r)=\emptyReg$ and $\mu_i(r_x)=v_x[i]$.
        Since $\cA$'s rule at step $i$ is enabled on the current $\A'$-letter, $\tau(\Phi_i)$ holds in $\cB$ at step $i$.
        And for $i>1$, $\bigwedge_x\eta_x$ yields $v_x[i]=v_x[i-1]$ whenever $u_x[i]=\utail$, i.e. $\Eta(\vu[i])$ holds.
        Therefore $\cB$ can follow the same state sequence and accept $\Conv{w,\varpi}$.
        This proves the required language equality.
\end{proof}

To prove \cref{lem:formula-translation-desired-property}, which concerns the
sentence $F(\varphi_0)$, we establish the more technical
\cref{lem:formula-translation-fixed} below.
This lemma also captures the
behaviour of $G(\varphi,S)$ for subformulas $\varphi$ of $\varphi_0$, enabling an
inductive proof on the structure of $\varphi_0$.
For a subformula $\varphi$ of $\varphi_0$, we distinguish the following two
subsets of $\VarIni$:
\begin{align*}
    \X{\varphi} &\coloneqq \VarIni \cap \FreeFOVarOf{\varphi}\,, &
    \Y{\varphi} &\coloneqq \VarIni \cap \FOVarOf{\varphi} \setminus \FreeFOVarOf{\varphi}\,.
\end{align*}
Fix $L \in \N$ and a scope $I \subseteq \rangeOne{L}$.
Let $\vu\colon A \to B^L$ and $\vcu\colon A' \to B^{\size{I}}$ be families of
words, where $A' \subseteq A$.
We say that \emph{$\vu$ is consistent with $\vcu$ on $I$} if
$u_x[I] = \check{u}_x$ for every $x \in A'$.
\begin{example}
    Let $X = \set{\ttx, \tty, \ttz}$, $X' = \set{\tty, \ttz}$, and
    $I = \range{3}{5}$. Let $\v{u}$ and $\v{\check{u}}$ be the following families:
    \begin{align*}
        \Conv{\vcu} &= {\footnotesize
        {\begin{matrix}
             \check{u}_\tty\!\colon\! \\
             \check{u}_{\ttz}\!\colon\!
        \end{matrix}}\!
        \underset{1}{
            \begin{bmatrix}
                \color{red}\ttd \\
                \color{red}\ttg
            \end{bmatrix}
        }\!
        \underset{2}{
            \begin{bmatrix}
                \color{red}\tte \\
                \color{red}\tth
            \end{bmatrix}
        }\!
        \underset{3}{
            \begin{bmatrix}
                \color{red}\ttf \\
                \color{red}\tti
            \end{bmatrix}
        }}\,,
        &
        \Conv{\vu} &= {\footnotesize
        {\begin{matrix}
             u'_\ttx\!\colon\!\\
             u'_\tty\!\colon\!\\
             u'_\ttz\!\colon\!
        \end{matrix}}\!
        \underset{1}{
            \begin{bmatrix}
                \tta \\
                \ttb \\
                \ttc
            \end{bmatrix}
        }\!
        \underset{2}{
            \begin{bmatrix}
                \tta \\
                \ttb \\
                \ttc
            \end{bmatrix}
        }\!
        \smash{\underbrace{
            \underset{3}{
                \begin{bmatrix}
                    \tta \\
                    \color{red}\ttd \\
                    \color{red}\ttg
                \end{bmatrix}
            }\!
            \underset{4}{
                \begin{bmatrix}
                    \tta \\
                    \color{red}\tte \\
                    \color{red}\tth
                \end{bmatrix}
            }\!
            \underset{5}{
                \begin{bmatrix}
                    \tta \\
                    \color{red}\ttf \\
                    \color{red}\tti
                \end{bmatrix}
            }}_{\normalsize\range{3}{5}}}\!
        \underset{6}{
            \begin{bmatrix}
                \tta \\
                \ttb \\
                \ttc
            \end{bmatrix}
        }}\,.
    \end{align*}
    Then $\vu$ is consistent with $\vcu$ on $I$.
\end{example}

\begin{lemma}
    \label{lem:formula-translation-fixed}
    Fix a subformula $\varphi$ of $\varphi_0$ and an interpretation
    $\fI = (\modelConv{w, \varpi}, \nu, I)$, where $I$ is nonempty.
    Let $L = \length{w}$.
    Fix a family of super-scopes $(I_x)_{x \in \X{\varphi}}$
    such that $I \subseteq I_x$ and $\nu(y) \in I_x$ for every $x \in \X{\varphi}$ and $y \in \VarOther \cap \FreeFOVarOf{\varphi}$.
    Let $S$ be a fresh SO variable, and let $\fI_{\vu, \vv}$ be the interpretation $(\modelConv{w, \vu, \varpi, \vv}, \nu[S \gets I])$ of $\SMSO[\Sigma', \A']$.
    Define two properties $P_{\cX}$ and $P_{\cY}$:
    \begin{itemize}
        \item
        \makebox[20mm][l]{$P_{\cX}(\varphi, \vu, \vv)$} holds iff
        \makebox[59mm][c]{$u_x[I_x] \in {\ubegin} {\utail}^*$ and $v_x[I_x] = \varpi[\nu(x)]^{\size{I_x}}$}
        for every $x \in \X{\varphi}$,
        \item \makebox[20mm][l]{$P_{\cY}(\varphi, \vcu, \vcv)$} holds iff
        \makebox[59mm][c]{$\modelConv{w[I], \vcu, \varpi[I], \vcv} \models \eta_y$}
        for every $y \in \Y{\varphi}$.
    \end{itemize}
    Then the following two claims hold:
    \begin{enumerate}[label=\textup{{(C\arabic*)}}, ref=\textup{(C\arabic*)}, leftmargin=*,topsep=\smallskipamount, partopsep=0pt]
        \item \label[claim]{claim:left-to-right}
        If $\fI \models \varphi$, then there exist infix families
        $\vcu\colon \Y{\varphi} \to \B^{\size{I}}$ and
        $\vcv\colon \Y{\varphi} \to \A^{\size{I}}$ satisfying
        $P_\cY(\varphi, \vcu, \vcv)$ such that, for all
        $\vu\colon \VarIni \to \B^L$ and $\vv\colon \VarIni \to \A^L$ that are
        consistent with $\vcu$ and $\vcv$ on $I$, if
        $P_{\cX}(\varphi, \vu, \vv)$ holds then
        $\fI_{\vu, \vv} \models G(\varphi, S)$.
        \item \label[claim]{claim:right-to-left}
        If $\fI_{\vu, \vv} \models G(\varphi, S)$ and the properties
        $P_{\cX}(\varphi, \vu, \vv)$ and $P_{\cY}(\varphi, \vu[I], \vv[I])$ hold
        for some $\vu\colon \VarIni \to \B^L$ and $\vv\colon \VarIni \to \A^L$,
        then $\fI \models \varphi$.
    \end{enumerate}
\end{lemma}

Note that in \cref{lem:formula-translation-fixed} we always evaluate translated
formulas $G(\varphi, S)$ over the ambient scope $\PosOf{w}$, and so we omit the
interval component in the corresponding interpretations.
We first explain how \cref{lem:formula-translation-fixed} implies
\cref{lem:formula-translation-desired-property}.

\begin{proof}[Proof of \cref{lem:formula-translation-desired-property} assuming \cref{lem:formula-translation-fixed}.]
    Fix an arbitrary word $\Conv{w, \varpi} \in (\Sigma \times \A)^*\!$.
    The claim is immediate when $w$ is empty: indeed, $\functionT{Empty}{}$ holds
    if and only if $\PosOf{w} = \emptyset$.
    Assume henceforth that $w$ is nonempty.
    We must show that
    \begin{align*}
        \Conv{w, \varpi} \in \lang{\varphi_0}
        &\iff
        \exists{\vcu, \vcv}.
        \Conv{{w, \vcu, \varpi, \vcv}} \in \lang[\big]{\exists{S}. G(\varphi_0, S) \land \forall{y} S(y) \land \textstyle\bigwedge_{x \in \VarIni}\eta_x}\,.
    \end{align*}
    By definition of language membership on the left-hand side, and since
    $I_0 \coloneqq \PosOf{w}$ is the unique set of positions satisfying
    $\forall{y} S(y)$, it suffices to show that
    \begin{align*}
        \modelConv{w, \varpi}, \nu_0, I_0 \models \varphi_0
        &\iff
        \exists{\vcu, \vcv}.
        \modelConv{{w, \v{\check u}, \varpi, \v{\check v}}}, \nu_0[S \gets I_0] \models G(\varphi_0, S) \land \textstyle\bigwedge_{x \in \VarIni}\eta_x\,,
    \end{align*}
    where $\nu_0$ is the empty valuation.
    Note that $\X{\varphi_0} = \emptyset$ and $\Y{\varphi_0} = \VarIni$.
    The desired equivalence follows directly from
    \cref{lem:formula-translation-fixed} with $I = I_0$.

    \proofDir{Direction ``$\bm{\Rightarrow}$''.}
By \cref{claim:left-to-right}, we obtain $\vcu, \vcv$ satisfying
$P_\cY(\varphi_0, \vcu, \vcv)$.
Any families $\vu, \vv$ consistent with $\vcu, \vcv$ on $I_0$ must coincide
with $\vcu, \vcv$, respectively.
Moreover, $P_\cX(\varphi_0, \vcu, \vcv)$ holds vacuously since
$\X{\varphi_0} = \emptyset$.
Hence $\modelConv{w, \vcu, \varpi, \vcv}, \nu_0[S \gets I_0] \models
G(\varphi_0, S)$, as required.

\proofDir{Direction ``$\bm{\Leftarrow}$''.}
Fix $\vcu, \vcv$ such that
$\modelConv{w, \vcu, \varpi, \vcv}, \nu_0[S \gets I_0] \models
G(\varphi_0, S) \land \textstyle\bigwedge_{x \in \VarIni}\eta_x$.
Then $P_\cY(\varphi_0, \vcu, \vcv)$ holds by definition, while
$P_\cX(\varphi_0, \vcu, \vcv)$ holds vacuously since $\X{\varphi_0}=\emptyset$.
By \cref{claim:right-to-left}, we conclude that
$\modelConv{w, \varpi}, \nu_0, I_0 \models \varphi_0$.
\end{proof}

\begin{proof}[\of{\cref{lem:formula-translation-fixed}}]
Fix a subformula $\varphi$ of $\varphi_0$ and an interpretation
$\fI = (\modelConv{w, \varpi}, \nu, I)$.
Fix a family of super-scopes $(I_x)_{x \in \X{\varphi}}$ such that
$I \subseteq I_x$ and $\nu(y) \in I_x$ for every
$x \in \X{\varphi}$ and $y \in \VarOther \cap \FreeFOVarOf{\varphi}$.
We proceed by induction on the structure of $\varphi$.
There are ten main cases, corresponding to the grammar rules of \SMSO.
In each case we prove
\cref{claim:left-to-right,claim:right-to-left}.

\proofCase{1}{{$\varphi \equiv R(x_1, x_2, \ldots, x_r)$}}
Let $\cV \coloneqq \set{x_1, x_2, \ldots, x_r}$.
Recall that $\size{\cV \cap \VarOther} \leq 1$.

\proofCase{1A}{$\cV \cap \VarOther = \set{z}$ for some $z$}
Let $p = \nu(z)$.
\begin{enumerate}[label=\textup{{(C\arabic*)}},leftmargin=*,topsep=\smallskipamount, partopsep=0pt]\item
Assume $\fI \models \varphi$.
By definition, $(\varpi[\nu(x_1)], \ldots, \varpi[\nu(x_r)]) \in R$.
Note that $\Y{\varphi} = \emptyset$.
Let $\vcu, \vcv$ be the unique empty families.
Then $P_\cY(\varphi, \vcu, \vcv)$ holds vacuously.
Take any $\vu, \vv$ satisfying $P_\cX(\varphi, \vu, \vv)$.
These families are trivially consistent with $\vcu$ and $\vcv$ on $I$.
It remains to show that, for every $i$, the term $t_{i,z}$ in $G(\varphi,S)$
evaluates to $\varpi[\nu(x_i)]$.
This is immediate when $x_i \equiv z$.
Otherwise, $t_{i,z}$ evaluates to $v_{x_i}[p]$.
Since $p \in I_{x_i}$, property $P_\cX$ yields
$v_{x_i}[p] = \varpi[\nu(x_i)]$, as required.

\item
Fix $\vu, \vv$ such that $P_\cX(\varphi, \vu, \vv)$,
$P_\cY(\varphi, \vu[I], \vv[I])$, and $\fI_{\vu, \vv} \models G(\varphi, S)$.
As in \cref{claim:right-to-left}, we obtain
$\varpi[\nu(x_i)] = v_{x_i}[p]$ for each $i$, and hence $\fI \models \varphi$.

\end{enumerate}
\proofCase{1B}{$\cV \cap \VarOther = \emptyset$}
\begin{enumerate}[label=\textup{{(C\arabic*)}},leftmargin=*,topsep=\smallskipamount, partopsep=0pt]\item
Assume that $\fI \models \varphi$.
Let $\vcu, \vcv$ be the unique empty infix families indexed by
$\Y{\varphi}$, and let $\vu, \vv$ be any word families that are consistent
with them on $I$ and satisfy $P_\cX(\varphi, \vu, \vv)$.
Then $P_\cY(\varphi, \vcu, \vcv)$ holds vacuously.
We must show that
$\fI_{\vu, \vv} \models \exists{z}\, (z \in S \land R(t_{1,z}, \ldots, t_{r,z}))$.
Choose any $p \in I$ (recall that $I \neq \emptyset$).
By the same argument as in Case~1A, we have
$\modelConv{w, \vu, \varpi, \vv}, \nu[z \gets p, S \gets I] \models
z \in S \land R(t_{1,z}, \ldots, t_{r,z})$.
Hence $\fI_{\vu, \vv} \models G(\varphi, S)$.

\item
Fix $\vu, \vv$ such that $\fI_{\vu, \vv} \models G(\varphi, S)$ and the
properties $P_\cX(\varphi, \vu, \vv)$ and
$P_\cY(\varphi, \vu[I], \vv[I])$ hold.
By definition, there exists $p \in I$ such that
$\modelConv{w, \vu, \varpi, \vv}, \nu[z \gets p, S \gets I] \models
R(t_{1,z}, \ldots, t_{r,z})$.
By the same argument as in Case~1A, we conclude that $\fI \models \varphi$.

\end{enumerate}
\proofCase{2}{$\varphi \equiv x < y$} Immediate.

\begin{enumerate}[label=\textup{{(C\arabic*)}},leftmargin=*,topsep=\smallskipamount, partopsep=0pt]\item
If $\fI \models x < y$, then $\nu(x) < \nu(y)$, and thus
$\fI_{\vu, \vv} \models x < y$ for any $\vu, \vv$ consistent with the unique
trivial infix families $\vcu, \vcv$ indexed by $\Y{\varphi} = \emptyset$.

\item
If $\fI_{\vu, \vv} \models x < y$ for some $\vu, \vv$, then $\fI \models x < y$.

\end{enumerate}
\proofCase{3}{$\varphi \equiv a(x)$} Immediate.

\begin{enumerate}[label=\textup{{(C\arabic*)}},leftmargin=*,topsep=\smallskipamount, partopsep=0pt]\item
If $\fI \models a(x)$, then $w[\nu(x)] = a$, and thus
$\fI_{\vu, \vv} \models w[x] = a$ for any $\vu, \vv$ consistent with the
unique trivial infix families $\vcu, \vcv$ indexed by $\Y{\varphi} = \emptyset$.

\item
If $\fI_{\vu, \vv} \models w[x] = a$ for some $\vu, \vv$, then $\fI \models a(x)$.

\end{enumerate}
\proofCase{4}{$\varphi \equiv X(y)$} Immediate.

\begin{enumerate}[label=\textup{{(C\arabic*)}},leftmargin=*,topsep=\smallskipamount, partopsep=0pt]\item
If $\fI \models X(y)$, then $\nu(y) \in \nu(X)$, and thus
$\fI_{\vu, \vv} \models X(y)$ for any $\vu, \vv$ consistent with the unique
trivial infix families $\vcu, \vcv$ indexed by $\Y{\varphi} = \emptyset$.

\item
If $\fI_{\vu, \vv} \models X(y)$ for some $\vu, \vv$, then $\fI \models X(y)$.

\end{enumerate}
\proofCase{5}{$\varphi \equiv \psi \land \gamma$}
Note that $\Y{\varphi} = \Y{\psi} \disjointUnion \Y{\gamma}$.

\begin{enumerate}[label=\textup{{(C\arabic*)}},leftmargin=*,topsep=\smallskipamount, partopsep=0pt]\item
Assume $\fI \models \psi \land \gamma$. Then $\fI \models \psi$ and
$\fI \models \gamma$.
By \cref{claim:left-to-right} of the induction hypothesis, there exist infix
families $\vcup, \vcvp$ satisfying $P_\cY(\psi, \vcup, \vcvp)$ such that, for
every $\v{u'}, \v{v'}$ consistent with them on $I$ and satisfying
$P_\cX(\psi, \v{u'}, \v{v'})$, we have $\fI_{\v{u'}, \v{v'}} \models G(\psi, S)$.
Likewise, there exist $\vcupp, \vcvpp$ witnessing the analogous statement
for $G(\gamma,S)$.
Let $\vcu = \vcup \disjointUnion \vcupp$ and $\vcv = \vcvp \disjointUnion \vcvpp$.
Then $P_\cY(\varphi, \vcu, \vcv)$ holds.
Take any $\vu, \vv$ satisfying $P_\cX(\varphi, \vu, \vv)$ and consistent with
$\vcu, \vcv$ on $I$.
Since $\vu, \vv$ are also consistent with $\vcup, \vcvp$ on $I$, and satisfy
$P_\cX(\psi, \vu, \vv)$ and $P_\cY(\psi, \vu[I], \vv[I])$, we obtain
$\fI_{\vu, \vv} \models G(\psi, S)$.
Similarly, $\vu, \vv$ are consistent with $\vcupp, \vcvpp$ on $I$ and satisfy
$P_\cX(\gamma, \vu, \vv)$ and $P_\cY(\gamma, \vu[I], \vv[I])$, so
$\fI_{\vu, \vv} \models G(\gamma, S)$.
Hence $\fI_{\vu, \vv} \models G(\varphi, S)$.

\item
Let $\vu, \vv$ be such that $P_\cX(\varphi, \vu, \vv)$ and
$P_\cY(\varphi, \vu[I], \vv[I])$ hold and $\fI_{\vu, \vv} \models G(\varphi, S)$.
Then $\fI_{\vu, \vv} \models G(\psi, S)$ and $\fI_{\vu, \vv} \models G(\gamma, S)$.
Moreover, $P_\cX(\psi, \vu, \vv)$ and $P_\cY(\psi, \vu[I], \vv[I])$ hold, and
similarly for $\gamma$.
By \cref{claim:right-to-left} of the induction hypothesis, we have
$\fI \models \psi$ and $\fI \models \gamma$, hence $\fI \models \psi \land \gamma$.

\end{enumerate}
\proofCase{6}{$\varphi \equiv \psi \lor \gamma$}
The proof is analogous to Case~5, with $\land$ replaced by $\lor$.

\proofCase{7}{$\varphi \equiv \neg \psi$}
Observe that $\Y{\psi} = \Y{\varphi} = \emptyset$, since no top-level
variables can be quantified under negation.

\begin{enumerate}[label=\textup{{(C\arabic*)}},leftmargin=*,topsep=\smallskipamount, partopsep=0pt]\item
Assume that $\fI \models \neg\psi$.
Let $\vcu, \vcv$ be the unique empty infix families; then
$P_\cY(\varphi, \vcu, \vcv)$ holds vacuously.
Towards a contradiction, suppose that
$\fI_{\vu, \vv} \models G(\psi, S)$ for some $\vu, \vv$ satisfying
$P_\cX(\varphi, \vu, \vv)$ and trivially consistent with $\vcu, \vcv$.
Then \cref{claim:right-to-left} yields $\fI \models \psi$, contradicting the
assumption.

\item
Assume that $\fI_{\vu, \vv} \models \neg G(\psi, S)$ for some $\vu, \vv$ such
that $P_\cX(\varphi, \vu, \vv)$ and $P_\cY(\varphi, \vu[I], \vv[I])$ hold.
Towards a contradiction, suppose that $\fI \models \psi$.
By \cref{claim:left-to-right} of the induction hypothesis, for the unique
empty families $\vcu, \vcv$ (which are consistent with $\vu,\vv$), we would
have $\fI_{\vu, \vv} \models G(\psi, S)$, contradicting the assumption.

\end{enumerate}
\proofCase{8}{$\varphi \equiv \exists{X}\psi$} Immediate.

\begin{enumerate}[label=\textup{{(C\arabic*)}},leftmargin=*,topsep=\smallskipamount, partopsep=0pt]\item
Assume that $\fI \models \exists{X} \psi$.
Then there exists $P \subseteq I$ such that
$\modelConv{w, \varpi}, \nu[X \gets P], I \models \psi$.
By \cref{claim:left-to-right} of the induction hypothesis, there are some infix families $\vcu, \vcv$ satisfying $P_\cY(\psi, \vcu, \vcv)$ such that for every $\vu, \vv$ consistent with them on $I$ and satisfying $P_\cX(\psi, \vu, \vv)$, satisfaction relation $\modelConv{w, \vu, \varpi, \vv}, \nu[X \gets P, S \gets I] \models G(\psi, S)$ holds.
We reuse $\vcu, \vcv$; fix arbitrary $\vu, \vv$ consistent with them on $I$ and such that $P_\cX(\varphi, \vu, \vv)$.
Observe that $P_\cY(\varphi, \vcu, \vcv)$ and $P_\cX(\psi, \vu, \vv)$ hold, and therefore $\modelConv{w, \vu, \varpi, \vv}, \nu[X \gets P, S \gets I] \models X \subseteq S$.
By definition, $\fI_{\vu, \vv} \models G(\varphi, S)$ holds.
\item
Assume that $\fI_{\vu, \vv} \models G(\varphi, S) \equiv \exists{X}. X \subseteq S \land G(\psi, S)$ for some $\vu, \vv$ satisfying $P_\cX(\varphi, \vu,\vv)$ and $P_\cY(\varphi, \vu[I],\vv[I])$.
By definition, there is some $P \subseteq \PosOf{w}$ such that $\modelConv{w, \vu, \varpi, \vv}, \nu[X \gets P, S \gets I] \models X \subseteq S \land G(\psi, S)$.
It is easy to see that $P_\cX(\psi, \vu,\vv)$ and $P_\cY(\psi, \vu[I],\vv[I])$.
Observe that $P \subseteq I$.
By \cref{claim:right-to-left} of the induction hypothesis we get $\modelConv{w, \varpi}, \nu[X \gets P], I \models \psi$, and, consequently, $\fI \models \varphi$.

\end{enumerate}
\proofCase{9}{$\varphi \equiv \exists{x}\psi$}

\proofCase{9A}{$x \not\in \VarIni$} Immediate.
Note that $\Y{\varphi} = \emptyset$, otherwise $x$ would be top-level.
\begin{enumerate}[label=\textup{{(C\arabic*)}},leftmargin=*,topsep=\smallskipamount, partopsep=0pt]\item
Assume that $\fI \models \varphi$; by definition there is $p \in I$ such that $\modelConv{w, \varpi}, \nu[x \gets p], I \models \psi$.
By \cref{claim:left-to-right} of the induction hypothesis, there are $\vcu, \vcv$ satisfying $P_\cY(\psi, \vcu, \vcv)$ such that for every $\vu, \vv$ consistent with them on $I$ and satisfying $P_\cX(\psi, \vu, \vv)$ the satisfaction relation $\modelConv{w, \vu, \varpi, \vv}, \nu[x \gets p, S \gets I] \models G(\psi, S)$ holds, so by definition $\fI_{\vu, \vv} \models G(\varphi, S)$; this is the required claim.
\item
Fix $\vu, \vv$ such that $P_\cX(\varphi, \vu, \vv)$ and $P_\cY(\varphi, \vu[I], \vv[I])$ and $\fI_{\vu, \vv} \models G(\varphi, S)$.
By definition, there is some $p \in \PosOf{w}$ such that $\modelConv{w, \vu, \varpi, \vv}, \nu[x \gets p, S \gets I] \models x \in S \land G(\psi, S)$.
Note that this implies $p \in I$.
By \cref{claim:right-to-left} of the induction hypothesis, we get that $\modelConv{w, \varpi}, \nu[x \gets p], I \models \psi$, and by definition, $\fI \models \varphi$, as required.

\end{enumerate}
\proofCase{9B}{$x \in \VarIni$}
Note $\Y{\varphi} = \Y{\psi} \disjointUnion \set{x}$.
Here, $G(\varphi, S) \equiv \exists{x}. x \in S \land G(\psi,S) \land \theta_{x, S}$.
\begin{enumerate}[label=\textup{{(C\arabic*)}},leftmargin=*,topsep=\smallskipamount, partopsep=0pt]\item
Assume that $\fI \models \varphi$.
Then there exists $p \in I$ such that $\modelConv{w, \varpi}, \nu[x \gets p], I \models \psi$.
By \cref{claim:left-to-right} of the induction hypothesis with $I_x = I$, fix $\vcup, \vcvp$ satisfying $P_\cY(\psi, \vcup, \vcvp)$ such that for all $\vu, \vv$ consistent with them on $I$ and satisfying $P_\cX(\psi, \vu, \vv)$, we have $\modelConv{w, \vu, \varpi, \vv}, \nu[x \gets p, S \gets I] \models G(\psi, S)$.
Let $\vcu = \vcup \disjointUnion \set{ u_x \mapsto \ubegin \utail^{\size{I} - 1}}$ and
$\vcv = \vcvp \disjointUnion \set{ v_x \mapsto \varpi[p]^{\size{I}}}$.
Then $P_\cY(\varphi, \vcu, \vcv)$ holds.
Fix any $\vu, \vv$ consistent with $\vcu, \vcv$ satisfying $P_\cX(\varphi, \vu, \vv)$.
Let $\fI'_{\vu,\vv} = (\modelConv{w, \vu, \varpi, \vv}, \nu[x \gets p, S \gets I])$.
By construction of the infix tracks $\check{u}_x, \check{v}_x$, properties $P_\cX(\psi, \vu, \vv)$ and $P_\cY(\psi, \vcu, \vcv)$ hold as well.
Hence $\fI'_{\vu,\vv} \models G(\psi, S)$.
Moreover, by the choice of $\check{u}_x$ and $\check{v}_x$, we also have $\fI'_{\vu,\vv} \models \theta_{x,S}$.
Therefore $\fI_{\vu, \vv} \models G(\varphi, S)$.

\item
Fix $\vu, \vv$ such that $P_\cX(\varphi, \vu, \vv)$, $P_\cY(\varphi, \vu[I], \vv[I])$ and $\fI_{\vu, \vv} \models G(\varphi, S)$ hold.
Let $\fI'_{\vu,\vv,p} = (\modelConv{w, \vu, \varpi, \vv}, \nu[x \gets p, S \gets I])$.
By definition, there exists $p \in \PosOf{w}$ such that $\fI'_{\vu,\vv,p} \models x \in S \land G(\psi, S) \land \theta_{x,S}$.
In particular $p \in S$.
Since $\fI'_{\vu,\vv,p} \models \theta_{x,S}$, we have $v_x[p] = \varpi[p]$ and $u_x[I] \in \ubegin \utail^*$.
Using the assumption $\modelConv{\vu[I], \vv[I]} \models \eta_x$, a simple induction shows that $v_x[q] = \varpi[p]$ for every $q \in I$.
Hence $P_\cX(\psi, \vu, \vv)$ and $P_\cY(\psi, \vu[I], \vv[I])$ hold.
By \cref{claim:right-to-left} of the induction hypothesis, we obtain $\modelConv{w, \varpi}, \nu[x \gets p], I \models \psi$. Since $p \in I$, we conclude that $\fI \models \varphi$.

\end{enumerate}
\proofCase{10}{$\varphi \equiv \segment{X}\psi$}
Let $S'$ be a fresh SO variable.
Since $I \neq \emptyset$, we have $\Subscopes{I, \nu(X)} \neq \emptyset$.
Let $\Subscopes{I, \nu(X)} = \set{J_1, \ldots, J_k}$, enumerated in increasing order.
\begin{enumerate}[label=\textup{{(C\arabic*)}},leftmargin=*,topsep=\smallskipamount, partopsep=0pt]\item
Assume that $\fI \models \segment{X} \psi$.
Then $\modelConv{w, \varpi}, \nu, J_i \models \psi$ for every $i$.
By \cref{claim:left-to-right} of the induction hypothesis,
for each $i \in \rangeOne{k}$ there exist infix families $\vcui{i}, \vcvi{i}$ satisfying $P_\cY\bigl(\psi, \vcui{i}, \vcvi{i}\bigr)$ such that for every $\vui{i}, \vvi{i}$ consistent with them on $J_i$ and satisfying $P_\cX\bigl(\psi, \vui{i}, \vvi{i}\bigr)$, we have
$\modelConv{w, \vui{i}, \varpi, \vvi{i}}, \nu[S' \gets J_i] \models G(\psi, S')$.
Let $\vcu = \vcui{1} \vcui{2} \cdots \vcui{k}$ and $\vcv = \vcvi{1} \vcvi{2} \cdots \vcvi{k}$, where concatenation is applied pointwise to the corresponding families of tracks.
We claim that $P_\cY(\varphi, \vcu, \vcv)$ holds.
Indeed, every track in $\vcu$ starts with $\ubegin$, since every track in $\vcui{1}$ does so by $P_\cY(\psi, \vcui{1}, \vcvi{1})$.
Moreover, the implication $\check{u}_x[s+1] = \utail \rightarrow \check{v}_x[s] = \check{v}_x[s+1]$ holds whenever $s, s+1$ lie within a single block $\vcui{i}, \vcvi{i}$, respectively, for some $i$.
But since first letters of $\check{u}_x^{(i)}$ are all $\ubegin$, it also holds trivially on block boundaries as well.
Hence, $P_\cY(\varphi, \vcu, \vcv)$ holds.
Now take any $\vu, \vv$ consistent with $\vcu, \vcv$ on $I$ and satisfying $P_\cX(\varphi, \vu, \vv)$, and fix any $i \in \rangeOne{k}$.
Since $\vu, \vv$ are consistent with $\vcui{i}, \vcvi{i}$ on $J_i$ and satisfy $P_\cX(\psi, \vu, \vv)$, we have $\modelConv{w, \vu, \varpi, \vv}, \nu[S \gets I, S' \gets J_i] \models G(\psi, S')$, .
Therefore $\fI_{\vu, \vv} \models G(\varphi, S)$.

\item
Fix $\vu, \vv$ such that $\fI_{\vu,\vv} \models G(\varphi, S)$ and the
properties $P_\cX(\varphi, \vu, \vv)$ and $P_\cY(\varphi, \vu[I], \vv[I])$ hold.
Then, for every $i \in \rangeOne{k}$,
$\modelConv{w, \vu, \varpi, \vv}, \nu[S \gets I, S' \gets J_i] \models G(\psi, S')$.
It is immediate that $P_\cX(\psi, \vu, \vv)$ and
$P_\cY(\psi, \vu[J_i], \vv[J_i])$ hold.
By \cref{claim:right-to-left} of the induction hypothesis, we obtain
$\modelConv{w, \varpi}, \nu, J_i \models \psi$ for every $i$.
Hence $\fI \models \segment{X}\psi$, as required.\qedhere
\end{enumerate}
\end{proof}

        \subsection{From automata to formulas}
        \label{subsec:from-automata-to-formulas}
        In this section we prove the remaining direction of \cref{thm:logic-NRA-equivalence}, as formalised by the following:
\begin{lemma}
    \label{lem:from-NRA-to-formulas}
    Fix finite $\Sigma$ and atoms $\A$.
    For every weakly guessing $\cA \in \NRA[\Sigma, \A]$ there exists a formula $\varphi \in \SMSO[\Sigma, \A]$ such that $\lang{\cA} = \lang{\varphi}$.
\end{lemma}

\namedparagraph{Proof roadmap.}
We split the argument into:
\begin{itemize}
    \item an introduction of simple \NRA, a subclass equally expressive as full \NRA,
    \item a standard MSO encoding of a run skeleton, and
    \item a geometric recursion guided by the live intervals of registers that allows us to check all register constraints in $\SMSO$, making use of the $\segmentSymbol$ modality.
\end{itemize}

\namedparagraph{Simple NRAs.}
We first define a class of \emph{simple NRA} which enforces injective register valuations and no register shuffling.
This is a well-known construction already present in literature concerning other kinds of register machines~\cite[Thm.~1]{injective-registers}, and can be adapted to \NRAs in a straightforward way. Therefore, we state the lemma without the proof.

Let $\cA = (Q, \Sigma, \A, R, Q_\ini, Q_\fin, \delta) \in \NRA_k[\Sigma, \A]$.
We say that $\cA$ is \emph{simple}, if every transition rule $(p, \sigma, \varPhi, q) \in \delta$ has a register constraint of the form $\varPhi \equiv\varPhi' \land \textstyle\bigwedge_{i\neq j} \valpre{r}_j \neq \valpre{r}_i \neq \valpost{r}_j \neq \valpost{r}_i \land \bigwedge_{i} \valpre{r}_i \mathbin{\sim_i} \valpost{r}_i$ where ${\sim_i} \in \set{=, \neq}$.

Intuitively, a simple NRA does not
store the same data value in two different registers at the same time,
and never moves a data value between registers, and explicitly says whether $\valpre{r}_i = \valpost{r}_i$ or $\valpre{r}_i \neq \valpost{r}_i$ for every $i$.
\begin{fact}
    For every $\cA \in \NRA[\Sigma, \A]$ there exists a simple $\cA' \in \NRA[\Sigma, \A]$ recognising the same language, and if $\cA$ is weakly-guessing, then $\cA'$ is so, too.
\end{fact}

For the remainder of the section, let us fix $\cA=(Q,\Sigma,\A,\cR,Q_\ini,Q_\fin,\delta) \in \NRA[\Sigma, \A]$ that is simple and weakly guessing.
Fix an enumeration of $\cR=\set{r_1,\ldots,r_k}$.

\namedparagraph{Run skeleton.}
Let $m=\size{\delta}$ and fix an enumeration of transitions $\kappa:\delta\to\rangeOne{m}$.
We introduce second-order variables $\Delta_1,\ldots,\Delta_m$ and let
$\Run_\cA(\Delta_1,\ldots,\Delta_m)$ express that:
\begin{itemize}
    \item each position carries exactly one transition rule label,
    \item all transition rules are matching the $\Sigma$-letter,
    \item adjacent transition rule labels have matching states, and
    \item the first rule starts in $Q_\ini$, and the last ends in $Q_\fin$.
\end{itemize}
We specify these properties in a standard way:
\begin{align*}
    &\Run_\cA(\Delta_1, \Delta_2, \ldots, \Delta_m) \coloneqq \forall{i, j}. \\
    &\qquad\qquad\!\!\!\!
    \underbrace{\Big[
        \textstyle\bigvee_{\substack{t=(\argumentDot, \sigma, \argumentDot, \argumentDot) \in \delta}} \Delta_{\kappa(t)}(i) \land \sigma(i) \land
        \textstyle\bigwedge_{\substack{t' \in \delta\\t' \neq t}} \neg\Delta_{\kappa(t')}(i)
        \Big]}_{\substack{\textup{Every position has exactly one associated transition}\\\textup{rule and it matches the letter $\sigma \in \Sigma$.}}}
    \land {} \\
    &\qquad\qquad\!\!\!\!
    \underbrace{\Big[\makebox[15mm][r]{$\functionT{succ}{i,j}$} \rightarrow
    \textstyle\bigvee_{\substack{t = (q, \argumentDot, \argumentDot, r) \in \delta\\t' = (r, \argumentDot, \argumentDot, s) \in \delta}}
    \Delta_{\kappa(t)}(i) \land \Delta_{\kappa(t')}(j)\Big]}_{\substack{\textup{Consecutive transition rules have matching states.}}}
    \land {} \\
    &\qquad\qquad\!\!\!\!
    \underbrace{\Big[\makebox[15mm][r]{$\functionT{first}{j}$} \rightarrow \textstyle\bigvee_{\substack{t = (q, \argumentDot, \argumentDot, \argumentDot) \in \delta\\q \in Q_\ini}}
    \Delta_{\kappa(t)}(j) \Big]}_{\textup{First transition rule begins in $q \in Q_\ini$.}}
    \land
    \underbrace{\Big[\makebox[15mm][r]{$\functionT{last}{i}$} \rightarrow \textstyle\bigvee_{\substack{t = (\argumentDot, \argumentDot, \argumentDot, q) \in \delta\\q \in Q_\fin}}
    \Delta_{\kappa(t)}(i) \Big]}_{\textup{Last transition rule ends with $q \in Q_\fin$.}}\,.
\end{align*}
It remains to define $\RegOK(\Delta_1,\ldots,\Delta_m)\in\SMSO$ such that:
\begin{align*}
    \Conv{w,\varpi}\models \exists{\Delta_1}\cdots\exists{\Delta_m}\Big(\Run_\cA(\Delta_1,\ldots,\Delta_m)\land \RegOK(\Delta_1,\ldots,\Delta_m)\Big)\text{ iff }\Conv{w,\varpi}\in\lang{\cA}
\end{align*}

\namedparagraph{Remaining goal.}
Intuitively, the formula $\RegOK$ needs to verify that all the transition constraints hold by checking the word positions on which the relevant register values appear in the word $\varpi$.
This would be straightforward if we had unrestricted access to the universal quantifier.
However, that would lead us to using more than one nested variable in data atomic formulas.
For this reason, we need to use a much more elaborate construction making use of the $\segmentSymbol$ modality.

\namedparagraph{Live intervals induced by $\Delta_1,\ldots,\Delta_m$.}
Let $\cV' \coloneqq \set{\cur}\ \cup\ \set{\valpre{r_i},\valpost{r_i}\suchthat i\in\rangeOne{k}}.$ Fix a tie-breaking order $\prec$ on $\cV'$.
From the $\Delta$-labels, for each $x\in\cV'$ we MSO-define the set $\mathcal{I}_x$ of \emph{live intervals} of $x$:
maximal ranges $I=\range{a}{b}$ on which the value of $x$ is constant in the (unique) run described by the $\Delta$'s.
Write $\Live_x(a,b)$ for the MSO predicate ``$\range{a}{b}\in\mathcal{I}_x$''.

\begin{claim}[Live intervals partition $\PosOf{w}$.]
    Each $\mathcal{I}_x$ is a disjoint partition of $\PosOf{w}$
\end{claim}
\begin{claim}[Weak guessing witness]
    If $I\in\mathcal{I}_{\valpre{r_i}}$ or $\mathcal{I}_{\valpost{r_i}}$ carries the data value $\alpha \in \A$,
    then there is a position $p \in I$ such that $\varpi[p]=\alpha$.
\end{claim}

Defining $\Live_x(a,b)$ from $\Delta_i$ labels is routine but notationally heavy, therefore the exact formula was omitted.
It uses the MSO-definable ``change points'' induced by whether the chosen rule enforces $\valpre{r_i}\neq\valpost{r_i}$.

\namedparagraph{Pairing intervals.}
We define a concept of \emph{operating intervals} of type $(x,y)$.
Fix $x,y\in\cV'$.
We define $\Pairs(x,y) \subseteq \mathcal{I}_x\times\mathcal{I}_y$ to be the set of pairs $(I_x,I_y)$ such that
\begin{align*}
    I_x\cap I_y \neq \emptyset,
    \qquad
    \min I_x \leq \min I_y,
    \qquad
    \max I_x \leq \max I_y.
\end{align*}
For $(I_x,I_y)\in\Pairs(x,y)$ define the \emph{operating interval}
\begin{align*}
    J(I_x,I_y) \;\coloneqq\; I_x\cup I_y \;=\; \range{\min I_x}{\max I_y}.
\end{align*}
If $x=y$, then $\Pairs(x,x)$ reduces to $\set{(I_x,I_x)\suchthat I_x\in\mathcal{I}_x}$, and then
$J(I_x,I_x)=I_x$.

\begin{lemma}[Bijection induced by the type $(x, y)$]
    \label{lem:pairs-bijection}
    Assume $x\neq y$.
    If $(I_x,I_y)\in\Pairs(x,y)$ then $\max I_x\in I_y$ and $\min I_y\in I_x$.
    Consequently:
    \begin{enumerate}
        \item for each $I_x$ appearing in some pair, there is at most one $I_y$ with $(I_x,I_y)\in\Pairs(x,y)$;
        \item for each $I_y$ appearing in some pair, there is at most one $I_x$ with $(I_x,I_y)\in\Pairs(x,y)$;
        \item if we enumerate the pairs $(I_x^1,I_y^1),(I_x^2,I_y^2),\ldots$ in increasing order of $\min I_x^i$,
        then $\min I_y^1<\min I_y^2<\cdots$, and the $I_y^i$ are pairwise disjoint.
    \end{enumerate}
\end{lemma}

\begin{proof}
    Let $(I_x,I_y)\in\Pairs(x,y)$.
    Since $I_x\cap I_y\neq\emptyset$ we have $\min I_y \leq \max I_x$.
    Together with $\min I_x \leq \min I_y$ this implies $\min I_y\in I_x$.
    Similarly, $\max I_x \leq \max I_y$ and $\min I_y\leq \max I_x$ imply $\max I_x\in I_y$.

    For (1), if $(I_x,I_y),(I_x,I_y')\in\Pairs(x,y)$ then both $I_y$ and $I_y'$ contain $\max I_x$.
    Since $\mathcal{I}_y$ is a partition, $I_y=I_y'$.
    (2) is symmetric using $\min I_y\in I_x$.
    For (3), enumerate by increasing $\min I_x^i$.
    Then $I_x^i$ are disjoint intervals from a partition, so $\max I_x^i < \min I_x^{i+1}$.
    Because $\max I_x^i\in I_y^i$, the $y$-interval containing $\max I_x^i$ must end before the $y$-interval containing
    $\min I_x^{i+1}$ begins (again by partition/disjointness of $\mathcal{I}_y$), hence $I_y^i$ lies strictly before $I_y^{i+1}$.
\end{proof}

\begin{lemma}[Alternating disjointness]
    \label{lem:alternating-disjointness}
    Assume that $x\neq y$. Let us enumerate $\Pairs(x,y)$ as $(I_x^1,I_y^1),(I_x^2,I_y^2),\ldots$ by increasing $\min I_x^i$.
    Let $J_i\coloneqq J(I_x^i,I_y^i)=I_x^i\cup I_y^i$.
    Then $J_i\cap J_{i+2}=\emptyset$ for all $i$.
    Thus the operating intervals split into two families of pairwise disjoint intervals:
    $\set{J_1,J_3,\ldots}$ and $\set{J_2,J_4,\ldots}$.
\end{lemma}

See \cref{fig:operating-intervals} for an illustration of the concept of alternative disjointness of operating intervals.

\begin{proof}
    By \cref{lem:pairs-bijection}, both $(I_x^i)_i$ and $(I_y^i)_i$ are sequences of pairwise disjoint intervals, ordered by increasing $\min I_x^i$ and $\min I_y^i$,
    and $I_x^i\cap I_y^i\neq\emptyset$.
    Assume towards a contradiction that $J_i\cap J_{i+2}\neq\emptyset$.
    Then $\min I_x^{i+2}=\min J_{i+2}\leq \max J_i=\max I_y^i$.
    Since the $x$-intervals are ordered and disjoint, $\max I_x^{i+1}<\min I_x^{i+2}$, hence
    \begin{align*}
        \max I_x^{i+1} < \min I_x^{i+2} \leq \max I_y^i.
    \end{align*}
    Since the $y$-intervals are ordered and disjoint, $\max I_y^i < \min I_y^{i+1}$, hence
    \begin{align*}
        \max I_x^{i+1} < \max I_y^i < \min I_y^{i+1}.
    \end{align*}
    Therefore $I_x^{i+1}$ lies strictly to the left of $I_y^{i+1}$, so $I_x^{i+1}\cap I_y^{i+1}=\emptyset$,
    contradicting $(I_x^{i+1},I_y^{i+1})\in\Pairs(x,y)$.
\end{proof}

\namedparagraph{Using $\segmentSymbol$ to eliminate witness variables two-at-a-time.}
The $\SMSO$ condition forbids choosing witness positions \emph{as a function of} a universally quantified position $p$.
Instead we will cover the word by disjoint operating intervals, choose witnesses once per operating interval,
and recurse.
Fix a recursion stage with a partition $\cP\disjointUnion\cQ=\cV'$.
Intuitively, variables in $\cP$ already have fixed witness positions $(p_x)_{x\in\cP}$ in the current scope,
while variables in $\cQ$ remain to be witnessed.

\proofCase{1}{Base case}
If $\cQ=\emptyset$, we can check the guard at all positions with a single nested variable $p$:
\begin{align*}
    \Psi_{\cV',\emptyset}(\bar p)
    \;\coloneqq\;
    \forall{p}\ \bigwedge_{t\in\delta}\Big(\Delta_{\kappa(t)}(p)\rightarrow \widehat{\varPhi}_t(p,\bar p)\Big),
\end{align*}
where $\widehat{\varPhi}_t$ is obtained from the register constraint $\varPhi_t$ by replacing:
each register-value variable $x\in\cV_\cR$ by its witness position $p_x$ and $\cur$ by $p$.
Then every data atom contains at most one nested variable (namely $p$), so this is well-formed.

\proofCase{2}{Inductive step}
Assume $\cQ\neq\emptyset$.
For each ordered pair $(x,y)\in\cQ\times\cQ$ (including $x=y$), we will:
\begin{enumerate}
    \item define the set of operating intervals $J(I_x,I_y)$ relevant for the current stage;
    \item split them into two disjoint families (odd/even) using \cref{lem:alternating-disjointness} (or triviality for $x=y$);
    \item quantify split-point sets $X^0_{x,y},X^1_{x,y}$ that cut the current scope into subscopes containing precisely those operating intervals;
    \item use $\segment{X^b_{x,y}}$ to run a subformula once per operating interval, where we existentially choose witness positions for $x$ and $y$, thus eliminating two variables when $x \neq y$, or just one otherwise.
\end{enumerate}

We write $\Op_\cQ(x,y; a_x,b_x,a_y,b_y)$ for the MSO-definable predicate defining the set $\Pairs(x,y)$ for $x, y \in \cQ$:
\begin{itemize}
    \item $\Live_x(a_x,b_x)\land \Live_y(a_y,b_y)\land \range{a_x}{b_x}\cap\range{a_y}{b_y}\neq\emptyset$,
    \item $a_x\leq a_y$ and $b_x\leq b_y$, and
    \item $\range{a_x}{b_x}$ is left-extremal: for any other $a_z,b_z$ such that $\Live_z(a_z,b_z)$ and $a_y \in \range{a_z}{b_z}$, the left end $a_x \leq a_z$ (breaking ties with $\prec$),
    \item $\range{a_y}{b_y}$ is right-extremal: for any other $a_z,b_z$ such that $\Live_z(a_z,b_z)$ and $b_x \in \range{a_z}{b_z}$, the right end $a_z \leq a_y$ (breaking ties with $\prec$).
\end{itemize}
Again, this condition is MSO-definably, but notationally heavy, thus we omit the exact formula.
Now we can define $\Psi_{\cP,\cQ}(\bar p)$, where $\bar p$ is a tuple of variables from $\cP$, as:
\begin{align*}
    \Psi_{\cP,\cQ}(\bar p)
    \;\coloneqq\;
    \bigwedge_{(x,y)\in\cQ\times\cQ}\ \bigwedge_{b\in\set{0,1}}\
    \exists{X^b_{x,y}}\Big(
    \Split^{b}_{\cQ}(x,y,X^b_{x,y})
    \ \land\
    \segment{X^b_{x,y}}\ \Theta_{\cP,\cQ}(x,y)
    \Big),
\end{align*}
where:
\begin{itemize}
    \item $\Split^{b}_{\cQ}(x,y,X)$ says that the cutpoints $X$ are such that the induced subscopes contain (possibly as a strict subset) all $b$-parity
    operating intervals $J(I_x^i,I_y^i)$ arising from the extremal interval pairs for $(x,y)$, listed in order.
    This is justified by \cref{lem:alternating-disjointness} (and by disjointness of $\mathcal{I}_x$ for $x=y$).

    \item $\Theta_{\cP,\cQ}(x,y)$ is the per-subsco\-pe action:
    \begin{align*}
        \Theta_{\cP,\cQ}(x,y;\bar p)
        \coloneqq {}
        &\Big[
            \exists{a_x,b_x,a_y,b_y}\ \Op_\cQ(x,y;a_x,b_x,a_y,b_y)
            \Big]\rightarrow \Theta'_{\cP,\cQ}(x,y;\bar p)
    \end{align*}
    where for $x \neq y$ formula $\Theta'_{\cP,\cQ}(x,y;\bar p)$ is defined as
    \begin{align*}
        \Theta'_{\cP,\cQ}(x,y;\bar p) \coloneqq \exists{p_x}\exists{p_y}\Big(\underbrace{p_x\in\range{a_x}{b_x}}_{\substack{\text{Witness for $x$ in its}\\\text{live range $\range{a_x}{b_x}$}}}\land \underbrace{p_y\in\range{a_y}{b_y}}_{\substack{\text{Witness for $y$ in its}\\\text{live range $\range{a_y}{b_y}$}}}\land
        \Psi_{\cP\cup\set{x,y},\ \cQ\setminus\set{x,y}}(\bar p,p_x,p_y)\Big)\,,
    \end{align*}
    and when $x = y$ we define
    \begin{align*}
        \Theta'_{\cP,\cQ}(x,y;\bar p) \coloneqq \exists{p_x}\Big(p_x\in\range{a_x}{b_x}\land
        \Psi_{\cP\cup\set{x},\ \cQ\setminus\set{x}}(\bar p,p_x)\Big)\,.
    \end{align*}
\end{itemize}

\begin{figure}[htbp]
    \centering
    \begin{subfigure}[b]{\linewidth}
        \hspace{-1cm}\includegraphics[width=1.2\linewidth]{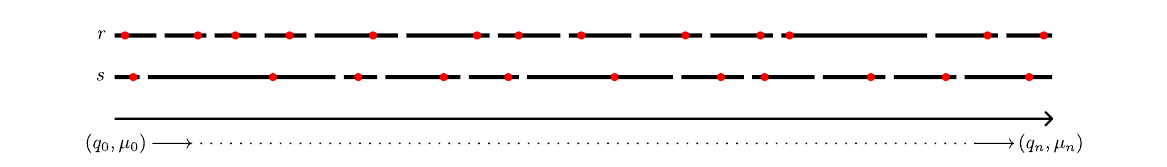}
        \caption{A run $\pi$ of $\cA$ and its two partitions into live intervals of registers $r$ and $s$.}
        \label{fig:pre}
    \end{subfigure}

    \vspace{0.5cm}
    \begin{subfigure}[b]{\linewidth}
        \hspace{-1cm}\includegraphics[width=1.2\linewidth]{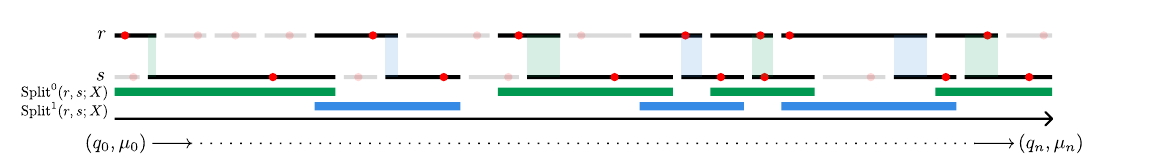}
        \caption{Operating intervals for type $(r, s)$. Intuitively, $r$-interval extends more to the left and $s$-interval to the right. Using two scope modalities in parallel, we can capture witness points (red dots) for every operating interval of that type.}
        \label{fig:rr}
    \end{subfigure}

    \vspace{0.5cm}

    \begin{subfigure}[b]{\linewidth}
        \hspace{-1cm}\includegraphics[width=1.2\linewidth]{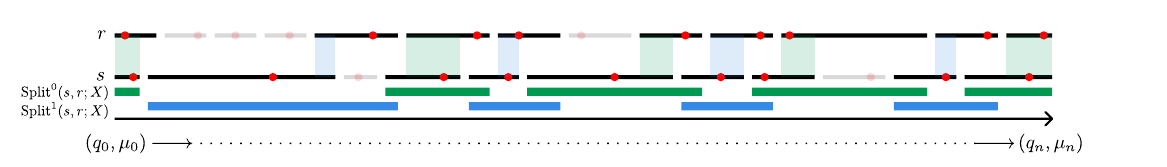}
        \caption{Operating intervals for type $(s, r)$. Intuitively, $s$-interval extends more to the left and $r$-interval to the right. Using two scope modalities in parallel, we can capture witness points (red dots) for every operating interval of that type.}
        \label{fig:ss}
    \end{subfigure}

    \vspace{0.5cm}
    \begin{subfigure}[b]{\linewidth}
        \hspace{-1cm}\includegraphics[width=1.2\linewidth]{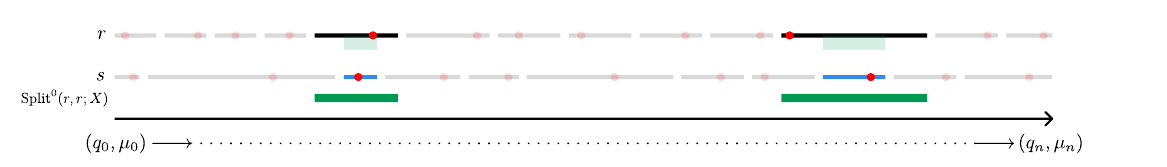}
        \caption{Operating intervals of type $(r, r)$. This degenerate case leads to eliminating only one register $r$.}
        \label{fig:rs}
    \end{subfigure}

    \vspace{0.5cm}
    \begin{subfigure}[b]{\linewidth}
        \hspace{-1cm}\includegraphics[width=1.2\linewidth]{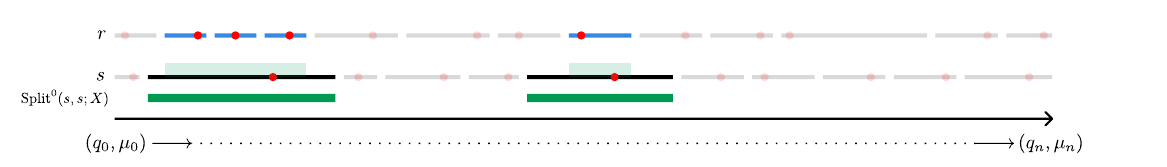}
        \caption{Operating intervals of type $(s, s)$. This degenerate case leads to eliminating only one register $s$.}
        \label{fig:sr}
    \end{subfigure}

    \caption{Formation of operating intervals for each of four types for two registers $\set{r, s}$.}
    \label{fig:operating-intervals}
\end{figure}

\noindent Finally, we define $\RegOK(\Delta_1,\ldots,\Delta_m) \;\coloneqq\; \Psi_{\emptyset,\cV'}$.

\namedparagraph{Correctness (sketch).}
\proofDir{Soundness.} ($\lang{\cA}\subseteq\lang{\varphi}$):
given an accepting run of $\cA$, interpret the $\Delta_i$ accordingly so $\Run_\cA$ holds.
Then, in each operating interval produced by $\segmentSymbol$, pick witness positions $p_x,p_y$
inside the corresponding live intervals that carry the correct register values on the data track
(existence by weak guessing). The base case checks the guard at every position, hence the formula holds.

\proofDir{Completeness.} ($\lang{\varphi}\subseteq\lang{\cA}$):
from $\Run_\cA$ obtain a rule label at every position (a run skeleton).
The recursion provides witness positions for all variables in $\cV'$ inside each operating interval,
and the base case enforces the translated guards everywhere.
Using simplicity (no shuffling + explicit $\valpre{r_i}\sim_i\valpost{r_i}$ + injectivity),
one reconstructs consistent register valuations along the skeleton, yielding an accepting run of $\cA$.
This proves $\lang{\cA}=\lang{\varphi}$.

    \section{Other proofs}
    \label{sec:other-proofs}

        \begin{proof}[Sketch of proof of \cref{lem:elimination-of-strong-guessing}]
    Elimination of strong guessing property for equality atoms $\A_{=}$ is proven as a solution to Exercise 7. \cite[p.~10,~207]{bojanczyk2019slightly}.
    We sketch the proof for dense order atoms $\A_{<} = (\Q, <, =)$.

    Fix an automaton $\cA \in \NRA[\Sigma, \A_{<}]$ for some finite $\Sigma$.
    Consider a run $\pi = (q_0, \mu_0) \xrightarrow{\sigma_1,\alpha_1} (q_1, \mu_1) \xrightarrow{\sigma_2, \alpha_2} \cdots \xrightarrow{\sigma_n, \alpha_n}$ of $\cA$ such that $\cA$ is performing strong guessing on $\pi$.
    Consider the set $\cI$ of all live intervals of all its registers;
    we depict them in \cref{fig:strong-guessing} horizontal line segments.
    The $y$-coordinate of an interval corresponds to the value stored in the register.
    Some of these intervals contain a position $i$ such that the register value is equal to the data value $\alpha_i$.
    We call such intervals weakly guessing; remaining intervals are strongly guessing.
    The idea is to construct an automaton $\cB$ which relaxes some of the transition guards corresponding to strongly guessing intervals.
    More specifically, some of the strict inequalities in a transition constraint $r_1 < r_2$ need to be weakened to $r_1 \leq r_2$ to allow all the strongly guessing intervals to be shifted vertically to reach the closest data value present in the word, making the guesses weak (cf.~\cref{fig:bottom}).
    It is clear that if $\cA$ has an accepting run, then $\cB$ will have one, too.
    For $\lang{\cB} \subseteq \lang{\cA}$, we need to ensure that a reverse operation, strengthening of the constraints, is possible.
    We do it by tracking the upper bound on how much we can shift the live interval in the opposite direction.
    This can be realised by adding a fresh $r'$ for every register $r$ of the original automaton.
\end{proof}

\begin{figure}[htbp]
    \centering
    \begin{subfigure}[b]{\linewidth}
        \includegraphics[width=\linewidth]{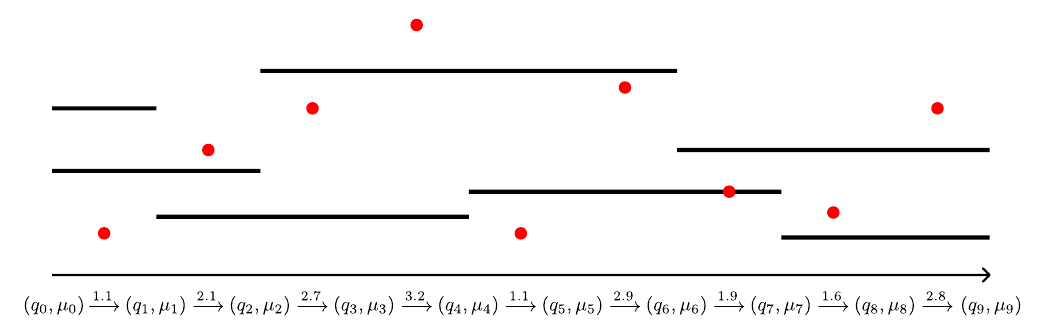}
        \caption{Live intervals in the run $\pi$ of $\cA$.}
        \label{fig:top}
    \end{subfigure}

    \vspace{0.5cm}

    \begin{subfigure}[b]{\linewidth}
        \includegraphics[width=\linewidth]{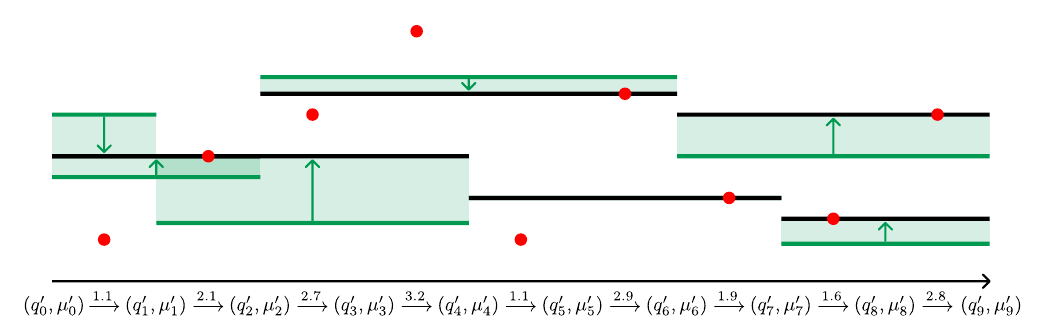}
        \caption{Live intervals in the run $\pi'$ of $\cB$.}
        \label{fig:bottom}
    \end{subfigure}

    \caption{Correspondence between a run $\pi$ of $\cA$ performing strong guesses and a run $\pi'$ of $\cB$ with only weak guesses.}
    \label{fig:strong-guessing}
\end{figure}
\end{document}